%% file: construct_lmcs.tex
\documentclass{lmcs} 


\usepackage{hyperref}

\usepackage[backend=biber,
  style=alphabetic,
]{biblatex}

\usepackage{mathrsfs}

\usepackage{newfile}
\usepackage{tikz}
\usetikzlibrary{babel,automata, positioning, arrows}

\usepackage{bussproofs}
\usepackage{braket}
\usepackage{stmaryrd}

\usepackage[inline]{enumitem}

\usepackage[capitalise]{cleveref}

\usepackage{float}
\usepackage{tcolorbox}
\usepackage[obeyDraft]{todonotes}
\usepackage[ruled,linesnumbered]{algorithm2e}

\input{defs.tex}

\newcommand{\takeout}[1]{}

\newboolean{proofsinappendix}
\setboolean{proofsinappendix}{false}

\newboolean{showappendix}
\setboolean{showappendix}{true}

\makeatletter
\ifthenelse{\boolean{proofsinappendix}}{%
  \newoutputstream{proofstream}
  \openoutputfile{\jobname-proofs.out}{proofstream}
  \NewEnviron{proofappendix}[1]{%
	\addtostream{proofstream}{\csname #1*\endcsname}%
    \addtostream{proofstream}{\noexpand\begin{proof}}%
    \addtostream{proofstream}{\noexpand\takeout{
      WARNING: this file is auto generated from the main tex file!
      Do not edit manually!
    }}
    \expandafter\addtostream{proofstream}{\expandonce{\BODY}}%
    \addtostream{proofstream}{\noexpand\end{proof}}%
    \addtostream{proofstream}{}%
  }%
}{%
  \newenvironment{proofappendix}[1]{%
    \begin{proof}
  }{
    \end{proof}
  }
}
\makeatother

\begin{document}

\title[Constructing Witnesses for Lower Bounds on Behavioural Distances\rsuper*]%
{Constructing Witnesses for Lower Bounds\\on Behavioural Distances}%
  \thanks{This work has benefitted from Dagstuhl Seminar 24432: Behavioural Metrics and Quantitative Logics.}

\titlecomment{{\lsuper*}This is an extended version of a paper that appeared in the proceedings of CSL 2026.}

\author[R.~Turkenburg]{Ruben Turkenburg\lmcsorcid{0000-0001-7336-9405}}[a]
\author[H.~Beohar]{Harsh Beohar\lmcsorcid{0000-0001-5256-1334}}[b]
\author[F.~van Breugel]{Franck van Breugel\lmcsorcid{0009-0002-7320-1527}}[c]
\author[C.~Kupke]{Clemens Kupke\lmcsorcid{0000-0002-0502-391}}[d]
\author[J.~Rot]{Jurriaan Rot\lmcsorcid{0000-0002-1404-6232}}[a]

\address{Radboud University, Netherlands}
\email{ruben.turkenburg@ru.nl,jurriaan.rot@ru.nl}
\thanks{This research is partially supported by the NWO grant No.~OCENW.M20.053.}

\address{University of Sheffield, UK}
\email{h.beohar@sheffield.ac.uk}
\thanks{EPSRC Grant: EP/X019373/1 and Royal Society Grant: IES\textbackslash R3\textbackslash 223092}

\address{York University, Canada}
\email{franck@yorku.ca}
\thanks{Natural Sciences and Engineering Research Council of Canada}

\address{University of Strathclyde, UK}
\email{clemens.kupke@strath.ac.uk}
\thanks{Leverhulme Trust Research Project Grant RPG-2020-232}


%
%





\begin{abstract}
	Behavioural distances provide a robust alternative to notions of equivalence such as bisimilarity in the context of probabilistic transition systems.
	They can be defined as least fixed points, whose universal property allows us to exhibit upper bounds on the distance between states, showing them to be \emph{at most} some distance apart.

	In this paper, we instead consider the problem of bounding distances from below, showing states to be \emph{at least} some distance apart.
	Contrary to upper bounds, it is possible to reason about lower bounds inductively.
	We exploit this by giving an inductive derivation system for lower bounds on an existing definition of behavioural distance for labelled Markov chains.
	This is inspired by recent work on \emph{apartness} as an inductive counterpart to bisimilarity.
	Proofs in our system will be shown to closely match the behavioural distance by soundness and (approximate) completeness results.
	Additionally, a correspondence will be shown between the existence of a proof of a non-zero lower bound, and proofs of \emph{behavioural apartness} in a system we developed in earlier work.

	We further provide a constructive correspondence between our derivation system and formulas in a modal logic with quantitative semantics.
	This logic was used in recent work of Rady and van Breugel to construct evidence for lower bounds on behavioural distances.
	Our constructions provide smaller witnessing formulas in many examples.
\end{abstract}

\maketitle


\section{Introduction}
Bisimilarity is an important notion of equivalence in the study of state-transition systems. It is \emph{qualitative} in the sense that states are either considered equivalent, or not; there are no degrees of equivalence. When studying systems involving probabilistic transitions, such qualitative definitions are usually considered too strict; states may be \emph{inequivalent} or \emph{distinguishable} under bisimilarity despite their behaviour being difficult to distinguish by an observer (this problem was first noted by Giacalone, Jou, and Smolka~\cite{DBLP:conf/ifip2/GiacaloneJS90}).

To better capture the (in)equivalence of states, quantitative notions of \emph{behavioural distances/metrics} may be used~\cite{DBLP:conf/concur/DesharnaisGJP99,DBLP:journals/siglog/Breugel17,DBLP:journals/tcs/BreugelW05}.
These assign to each pair of states a number (e.g., in the interval \([0,1]\)) representing how close (or how far apart) their behaviours are.
Determining these distances has been studied algorithmically, with procedures developed for approximating the distance~\cite{DBLP:journals/lmcs/BreugelSW08}, and the exact computation of distances~\cite{DBLP:conf/fossacs/ChenBW12,DBLP:journals/lmcs/Bacci0LM17,DBLP:journals/lmcs/BacciBLMTB21,tang2018computing}.
The definition of behavioural distances as least or greatest fixed points (depending on the chosen ordering) also gives them a universal property yielding a (co)inductive proof principle~\cite{DBLP:journals/iandc/HermidaJ98,DBLP:journals/lmcs/BaldanBKK18,DBLP:journals/mscs/Bonchi0P23}. The corresponding proofs are of bounds on the distances showing states to be equivalent to some degree.
This is analogous to the qualitative proof technique of exhibiting some bisimulation containing a pair of states, thereby showing them to be bisimilar.

\paragraph{Apartness}
Orthogonally, interest in (qualitative) apartness of states has been growing as an inductive counterpart to bisimilarity~\cite{DBLP:journals/lmcs/GeuversJ21}.
Rather than defining when states behave the same, apartness defines when there is some observable difference between them.
A reason for interest in apartness is its inductive and potentially constructive nature. Indeed, this was the original motivation, going back to the school of Brouwer~\cite{heyting1966intuitionism}.
In the setting of state-based systems, this means giving some (finite) evidence or witness for a difference in behaviours. Think of, for example, a word which is accepted by one state of a finite automaton but not another, or a difference in probability of making a certain observation in probabilistic systems such as labelled Markov chains or Markov decision processes.

As in the case of bisimilarity, we would like these notions to be as robust as possible, making behavioural distances a clear area of interest.
In the quantitative setting, the dual of the existing coinductive proof methods allows us to obtain lower bounds on distances between states, i.e., we can show states to be at least a certain distance apart.
This type of bound has recently been explored in~\cite{DBLP:journals/lmcs/BaldanEKP23}. They define a measure of how much a candidate for the least fixed point can be decreased. If no such decrease is possible, we have a lower bound.

In this work, we take an alternative approach, based essentially on Kleene's chain construction of least fixed points~\cite{Kleene1952Introduction}. For the case of behavioural distances, this starts from an order-preserving functional, say \(\Gamma \colon \PMet_X \to \PMet_X\), on the space of pseudometric spaces on a set \(X\). To approach the least fixed point \(\mu \Gamma\) from below, we can start from the constant zero distance \(\bot\) and iteratively apply \(\Gamma\) giving the chain \(\bot \leq \Gamma(\bot) \leq \Gamma^2(\bot) \leq \ldots\). As is noted in~\cite{DBLP:journals/lmcs/BaldanEKP23}, fully applying \(\Gamma\) iteratively in this way is not a desirable means of obtaining bounds. Instead, we will develop an inductive derivation system allowing the construction of lower bounds for chosen pairs of states.

However, simply translating the definition of the behavioural distance to proof rules does not work; to obtain a usable derivation system we must show that proof steps need only consider direct successors of the involved states.
This means that finite derivations can be constructed also for systems with infinite state spaces.
We further reduce the work required for proofs, and thus also the size of the derivations, by showing that recursive proofs are only required for a subset of these successors.

The judgments proved in our system are of the form \(x \apart{\varepsilon} y\) for \(x,y\) states in an LMC and \(\varepsilon\) a rational in the interval \([0,1]\). To ensure a correct relation to the behavioural distance, we show two properties. First, soundness: if we can prove \(x \apart{\varepsilon} y\), then \(\varepsilon\) is a lower bound on the behavioural distance between \(x\) and \(y\). Second, a form of completeness which we call \emph{approximate} completeness.
Usual completeness with respect to the behavioural distance would state that any lower bound on the distance between states can be proved in the derivation system. However, in the spirit of apartness, we consider \emph{finite} evidence, which can only witness finite approximations of distances in general. We thus show that lower bounds can be derived with arbitrarily small error with respect to the true distance.

In earlier work~\cite{DBLP:conf/cmcs/TurkenburgBKR24}, we developed an inductive proof system for \emph{behavioural apartness}, the dual of coalgebraic behavioural equivalence.
While that work did not focus on LMCs, we gave a method to derive proof systems for a class of coalgebras which includes LMCs.
A common expectation in the study of behavioural distances and bisimilarity for probabilistic systems, is that behavioural distance zero and bisimilarity should coincide (see, e.g.,~\cite{DBLP:conf/concur/DesharnaisGJP99,DBLP:journals/tcs/FahrenbergL14}).
In the same vein, we show that a proof of behavioural apartness corresponds to the existence of a proof of some non-zero lower bound in the system we develop in this work.

\paragraph{Logics}
Evidence of differences in behaviour can also be given in the form of formulas in some modal logic. This is closely related to Hennessy-Milner type theorems, which show for a given logic that bisimilarity and logical equivalence coincide.
Such theorems have been shown for probabilistic bisimilarity and a logic with a modality parameterised by rational probabilities~\cite{DBLP:conf/lics/DesharnaisEP98,DBLP:journals/iandc/DesharnaisEP02,DBLP:conf/icalp/FijalkowKP17}.
Later, a quantitative analogue was shown by relating a real-valued logic introduced in~\cite{DBLP:conf/concur/DesharnaisGJP99} to behavioural distances~\cite{DBLP:journals/tcs/BreugelHMW07}.
The most interesting part in the qualitative case is the inclusion of logical equivalence in bisimilarity, also called expressiveness, which can dually be shown by giving, for each pair of non-bisimilar states, a formula which distinguishes them. Quantitatively, this means giving a formula for which the difference in interpretations matches the behavioural distance as closely as possible.

This correspondence of behavioural distances and modal formulas has been investigated for \emph{labelled Markov chains} (LMCs) in~\cite{DBLP:conf/fossacs/RadyB23}. In their terminology, a construction is given of formulas ``explaining'' the distance between states. Due to the chosen logic, and the possibility of infinite behaviours in LMCs induced by loops, this can not be done exactly. Instead, it is shown that for any finite approximation \(\Gamma^i(\bot)(x,y)\) of the distance of states as in the above chain, a formula can be constructed such that \(|\interp{\varphi}(x) - \interp{\varphi}(y)|\) (the difference in interpretations on states \(x\) and \(y\)) is equal to the approximation.

We finish by relating our new derivation system to the work of~\cite{DBLP:conf/fossacs/RadyB23} by showing that for any derivation, a formula in the modal logic of \emph{op.\ cit.} can be constructed which witnesses the same bound. These witnessing formulas are an improvement as they can be given for infinite state systems, and they will be smaller in many examples. Further, for any formula a proof tree witnessing the same lower bound can be given whose depth will be bounded by the number of operators in the formula. In case we measure the constructed proofs by the number of rule applications, the derivations will be larger in general, as they are dependent on the system and thus contain more information.
This is exactly what facilitates the aforementioned improvements; we see the steps which lead us to conclude a difference in behaviours which are otherwise somewhat hidden in the semantics of the logic.
Proofs also focus on pairs of states, so that at each step we see which states of a system are being used to exhibit a lower bound.

\paragraph{Contributions}
Summarising the contributions of the paper:
\begin{itemize}
	\item We define an inductive derivation system for lower bounds on behavioural distances in labelled Markov chains (\cref{sec:proofsys})
	\item We show the soundness and approximate completeness of the system with respect to the behavioural distance (\cref{sec:sound,sec:appcomp})
	\item We show a correspondence between proofs of some non-zero lower bound, and proofs of the earlier notion of behavioural apartness (\cref{sec:reltwovalapart})
	\item We show a constructive correspondence between proofs in our system and formulas of a modal logic (\cref{sec:logic,sec:ftop})
\end{itemize}
We illustrate all of the above with examples, including a system with an infinite state space which was not covered by this modal logic (\cref{ex:randwalk}).

\paragraph{Related Work}
The line of work focussing on proofs of apartness for state-based systems was (re)started by Geuvers and Jacobs~\cite{DBLP:journals/lmcs/GeuversJ21}, with further work on the relation to distinguishing formulas in~\cite{DBLP:conf/birthday/Geuvers22}.
A proof system for an apartness notion dual to coalgebraic behavioural equivalence has been given in~\cite{DBLP:conf/cmcs/TurkenburgBKR24}.\ruben{anonymised}
The current paper is an extended version of~\cite{DBLP:conf/csl/TurkenburgBBKR26}, building directly on these works.
Besides the inclusion of full proofs and more details and explanations, the current extended version contains a new part on the relation between the quantitative proof system that is the main topic of this paper, and the two-valued proof system for behavioural apartness from~\cite{DBLP:conf/cmcs/TurkenburgBKR24} (\cref{sec:reltwovalapart}); and a more elaborate discussion of the use of linear programming used in the completeness proof (\cref{sec:LProg}).

The coinductive proof principle in the context of behavioural distances has been explored coalgebraically in, e.g.,~\cite{DBLP:journals/entcs/Rutten98,DBLP:journals/acta/BonchiPPR17,DBLP:journals/lmcs/BaldanBKK18,DBLP:journals/mscs/Bonchi0P23}. In the greatest fixed point characterisation, the ordering is reversed compared to the definition we use in the rest of this work. Coinduction thus leads to upper bounds on distances under our definition.
An approach focussed on bounding greatest fixed points from above (but which dually bounds least fixed points from below as we will do) has more recently been given in~\cite{DBLP:journals/lmcs/BaldanEKP23}, as discussed above.
There is however no construction given of formulas demonstrating proved bounds.

The construction of distinguishing formulas has also been studied in the qualitative setting in, e.g.,~\cite{DBLP:conf/concur/0001G23,DBLP:conf/birthday/0001G24,DBLP:conf/cmcs/0001MS20,DBLP:journals/lmcs/WissmannMS22}.
We discuss these works further in \cref{sec:conclude}.

A more general account of bounding distances from above is the area of quantitative equational theories~\cite{DBLP:conf/lics/MardarePP16}, which has been applied to give a calculus for upper bounds on distances in Markov chains~\cite{DBLP:journals/lmcs/BacciBLM18} and regular expressions~\cite{DBLP:conf/icalp/Rozowski24}.

\paragraph*{Notation}
We will write, \(\D_\mathbb{Q}(X)\) for the set of finitely-supported rational distributions on the set \(X\). These are maps \(\mu \colon X \to [0,1] \cap \mathbb{Q}\) such that \(\supp(\mu) := \{ x \in X \mid \mu(x) \neq 0 \}\) is finite and \(\sum_{x \in X} \mu(x) = 1\). We may also write such distributions as formal sums: \(\sum_{x \in X} \mu(x) \ket{x}\), with the \(\ket{x}\) (``ket'') notation taken from~\cite{JacobsStructProb} and used simply to separate states and their associated probabilities. It can also be thought of as indicating a (weighted) sum of Dirac distributions. From now on, we will write \(\UIQ\) for \([0,1] \cap \mathbb{Q}\).

We denote by \(\PMet_X\) the set of (1-bounded) pseudometric spaces on a set \(X\), i.e., pairs \((X,d)\) with \(d \colon X \times X \to [0,1]\) a pseudometric.
We order the unit interval with the usual ordering of the reals, and pseudometrics inherit this ordering pointwise, so that \(d_1 \leq d_2\) iff \(\forall x,y \in X \ldotp d_1(x,y) \leq d_2(x,y)\).
The smallest element \(\bot\) is thereby the constant zero distance.
Further, for two pseudometric spaces \((X,d_X) \in \PMet_X\) and \((Y,d_Y) \in \PMet_Y\) a map \(f \colon (X,d_X) \to (Y,d_Y)\) will be assumed to be \emph{non-expansive} (also called \emph{1-Lipschitz}, \emph{short}, etc.), i.e., \(\forall x,y \in X \ldotp d_Y(f(x),f(y)) \leq d_X(x,y)\).
The Euclidean distance is denoted \(d_e \colon [0,1] \times [0,1] \to [0,1]\).

In the interest of space, we write \(\mu \cdot h\) for \(\sum_{x \in X} \mu(x) \cdot h(x)\) where \(\mu \colon X \to \UIQ\) is a distribution and \(h \colon X \to \mathbb{R}\) is an arbitrary function. This is motivated by viewing the distributions and functions as vectors indexed by their common domain, so that the operation is the vector dot product. In the sequel, we will often restrict \(h\) to maps into \(\UIQ\).

\section{Behavioural Distances on LMCs}\label{sec:behdist}
We start with the definition of the type of system which we study in the remainder of the paper: labelled Markov chains.

\begin{definition}
	A labelled Markov chain (LMC) consists of the following data:
	\begin{itemize}
		\item a set of states \(X\);
		\item a (non-empty) set of labels \(L\);
		\item a (finitely branching) probabilistic transition function \(\tau \colon X \to \D_\mathbb{Q}(X)\); and
		\item a labelling function \(l \colon X \to L\).
	\end{itemize}
\end{definition}

\begin{example}
	Let \(X = \{x, x_1, x_2\}\) and \(L = \{a,b\}\). We represent the LMC \((X,L,\tau,l)\) with \(\tau(x) = \frac{1}{2} \ket{x_1} + \frac{1}{2} \ket{x_2}, \tau(x_1) = 1 \ket{x_1}, \tau(x_2) = 1 \ket{x_2}\) and \(l(x) = l(x_1) = a, l(x_2) = b\) as: 
	
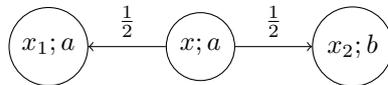
\begin{center}
		\begin{tikzpicture}[shorten >=0pt,node distance=2cm,
				on grid,auto,initial text=,every edge/.append style={bend angle=15}]
			\node[state] (x) {\slab{x}{a}} ;
			\node[state] (x1) [left=of x] {\slab{x_1}{a}} ;
			\node[state] (x2) [right=of x] {\slab{x_2}{b}} ;

			\path[->]
			(x) edge node[swap] {\(\frac{1}{2}\)} (x1)
			(x) edge node {\(\frac{1}{2}\)} (x2)

			;
		\end{tikzpicture}
	\end{center}
	We use the notation \slab{x}{a} for a state \(x \in X\) such that \(l(x) = a\). Further, any state with no outgoing edges is assumed to have a self-loop with probability \(1\).
\end{example}

We now recall a definition of the behavioural distance (henceforth written \(\bd\)) of states in an LMC as the least fixed point of a functional based on non-expansive maps.
This distinguishes two cases.
States having different labels should be maximally far apart, so they have distance \(1\).
The distance of states with the same label is then defined recursively, and can be seen as an optimisation problem.
Intuition for this problem is most often given in terms of its dual based on couplings under the Kantorovich-Rubinstein duality. The distance between distributions can in that setting be seen as the minimal cost of transporting one unit of mass from one distribution to the other, with the cost of transporting \(m\) probability mass between states at distance \(d\) being \(m \cdot d\).
In a more general form (as discussed in~\cite{villani2008optimal}) the distance below can be seen as the maximisation of profit, where we think of buying from one distribution and selling to the other, with the maps \(h\) representing buying/selling costs.

On small examples, these intuitions can may be used to determine the distance by examination.
On larger systems an approach, which we apply in \cref{sec:appcomp}, is to transform the supremum case in the definition of \(\Gamma\) into a (rational) linear program.
The distance computation can then be automated using existing solvers.
For further discussions of these distances, see for example~\cite{DBLP:journals/siglog/Breugel17, DBLP:conf/fossacs/RadyB23}. 
\begin{definition}
	For \(X\) a set, and \(\PMet_X\) the set of pseudometric spaces on \(X\), we define \(\Gamma \colon \PMet_X \to \PMet_X\) by
	\begin{align*}
		\Gamma(d)(x,y) = \begin{cases}
			                 1,                                                                           & \text{ if } l(x) \neq l(y), \\
			                 \sup_{h \colon (X,d) \to ([0,1],d_e)} (\tau(x) - \tau(y)) \cdot h, & \text{ o.w.}
		                 \end{cases}
	\end{align*}
	Then we define \(\bd := \lfp(\Gamma)\).
\end{definition}
Note that the least fixed point exists, because \(\PMet_X\) is a complete lattice, and \(\Gamma\) preserves the pointwise order on \(\PMet_X\).

It will be important for the correspondence results of later sections that the behavioural distance can be obtained as a countable supremum, namely the supremum over all finite applications of \(\Gamma\) to the constant zero distance.
A similar result for LMCs with non-determinism is shown already in~\cite[Sec.~3]{DBLP:conf/lics/DesharnaisJGP02}.
It can also be proved using the Kleene fixpoint theorem, or
\(\omega\)-(co)continuity of \(\Gamma\) as shown in~\cite{DBLP:journals/ipl/Breugel12}.
\begin{proposition}\label{prop:kleene}
	For any LMC \((X,L,\tau,l)\) and \(x,y \in X\), we have
	\begin{equation*}
		\bd(x,y) = \sup_{i < \omega} \Gamma^i(\bot)(x,y)
	\end{equation*}
\end{proposition}

\begin{example}\label{ex:LMC}
	Consider the following LMC:
	\begin{center}
		\begin{tikzpicture}[shorten >=0pt,node distance=2cm,
				on grid,auto,initial text=,every edge/.append style={bend angle=15}]
			\node[state] (x) {\slab{x}{a}} ;
			\node[state] (x1) [left=of x] {\slab{x_1}{a}} ;
			\node[state] (x2) [right=of x] {\slab{x_2}{b}} ;
			\node[state] (y) [right=7cm of x] {\slab{y}{a}} ;
			\node[state] (y1) [left=of y] {\slab{y_1}{a}} ;
			\node[state] (y2) [right=of y] {\slab{y_2}{b}} ;

			\path[->]
			(x) edge node[swap] {\(\frac{1}{2}\)} (x1)
			(x) edge node {\(\frac{1}{2}\)} (x2)
			(y) edge node[swap] {\(\frac{2}{5}\)} (y1)
			(y) edge node {\(\frac{3}{5}\)} (y2)
			;
		\end{tikzpicture}
	\end{center}
	Note that \(\bd(x_1,y_1) = \bd(x_2,y_2) = 0\) and \(\bd(x_1,y_2) = \bd(x_2,y_1) = 1\), which are the values given by \(\Gamma(\bot)\). The value \(\bd(x,y)\) is then \(\Gamma^2(\bot)(x,y)\) for which it can be shown that the supremum is achieved by the map \(h_0(z) = \textbf{if } z \in \{x_1,y_1\} \textbf{ then } 1 \textbf{ else } 0\) so that: 
	\begin{align*}
		\Gamma^2(\bot)(x,y) & = \sup_{h \colon (X,\Gamma(\bot)) \to ([0,1],d_e)} (\tau(x) - \tau(y)) \cdot h                                         \\
		                    & = (\tau(x) - \tau(y)) \cdot h_0                                                                                      \\
		                    & = \left(\frac{1}{2} - \frac{2}{5}\right) \cdot 1 - \left(\frac{1}{2} - \frac{3}{5}\right) \cdot 0 = \frac{1}{10}
	\end{align*}
	In terms of transportation costs, this \(\frac{1}{10}\) can be seen as the cost of transporting the \(\frac{1}{10}\) of mass from \(x_1\) to \(y_2\) which can not be transported to \(y_1\) due to its lack of capacity.
	Intuition based on profits is harder to give in this case, due to the asymmetry of our definition. It is easier to obtain the \(h\) in this example by seeing it as grouping the states into ``equivalence'' classes, spaced as far apart as possible while respecting the distance on \(X\) (in this case \(\Gamma(\bot)\)).
\end{example}

\section{Proof System}\label{sec:proofsys}

In this section, we define our derivation system for lower bounds on behavioural distances between states of an LMC.
The conclusions of the rules are of the form \(x \apart{\varepsilon} y\) which, as our soundness result will show, implies that \(\bd(x,y) \geq \varepsilon\), i.e., the behavioural distance between \(x\) and \(y\) is at least \(\varepsilon\).
The definition of \(\bd\) suggests two rules, one for each case in the definition of \(\Gamma\). The label case straightforwardly yields the rule
\begin{equation*}
	\AxiomC{\(l(x) \neq l(y)\)}
	\RightLabel{(lab)}
	\UnaryInfC{\(x \apart{1} y\)}
	\DisplayProof
\end{equation*}
In the recursive case, \(\bd(x,y)\) can be bounded from below by \((\tau(x) - \tau(y)) \cdot h\) for any non-expansive map \(h\) by definition.
However, it is not immediately clear that this can be done inductively, as we can not assume \(\bd\) to be known, and thus can not use it to choose a non-expansive \(h\).
Fortunately, as long as the system is sound with respect to the behavioural distance, it suffices to have a map \(h\) for which a kind of \emph{pairwise non-expansiveness} holds: for any \(x',y'\) we have \(|h(x') - h(y')| \leq \varepsilon\) for some \(\varepsilon\) such that we have proved \(x' \apart{\varepsilon} y'\). Soundness then implies that \(|h(x') - h(y')| \leq \varepsilon \leq \bd(x',y')\) for all \(x',y'\), which is exactly non-expansiveness of \(h\) with respect to \(\bd\).

Now we see two ways to build up proofs based on the idea of pairwise non-expansiveness:
(1) we start from a set of proved bounds, and search for a map \(h\) which satisfies the non-expansiveness property; or (2) we choose an \(h\) and then go about proving that it satisfies the required bounds.

In writing the corresponding rule, we choose a form which most closely fits approach (2), primarily because it makes the proof obligations clearer, in that we are required to prove precisely those bounds induced by the chosen \(h\).

A first version of such a rule is as follows:
\begin{equation*}
	\AxiomC{\(h \colon X \to [0,1]\)}
	\AxiomC{\(\forall x', y' \in X \ldotp x' \apart{|h(x') - h(y')|} y'\)}
	\AxiomC{\((\tau(x) - \tau(y)) \cdot h \geq \varepsilon\)}
	\TrinaryInfC{\(x \apart{\varepsilon} y\)}
	\DisplayProof
\end{equation*}
Note that the recursive proof obligations are of the exact bounds \(|h(x') - h(y')|\) rather than of bounds \(\varepsilon \geq |h(x') - h(y')|\) as in the above discussion.
This is primarily for notational convenience, and the two forms can be reconciled using the rule \emph{(weak)} which we include in \cref{def:proofrules} below.

However, in this form, the rule does not yet give a usable derivation system.
Namely, in case an LMC has an infinite state space, there will be infinitely many recursive proof obligations.
We will ensure that this is always finite, and even reduce the work required to construct proofs beyond this.
Further, the map \(h\) is so far \emph{real}-valued. This poses a problem both for our approximate completeness result, in which we need to be able to compute these maps, and for the construction from proofs to modal formulas whose interpretation will be rational-valued.

To remedy these issues, we adapt the above rule into our final \emph{(exp)} rule as follows:
\begin{itemize}
	\item we take \(h\) to be defined only on a finite subset of the state space thereby generating only finitely many proof obligations;
	\item we restrict the codomain of \(h\) to rationals;
	\item we reduce the number of recursive proof obligations further by not requiring proofs for those bounds which follow from reflexivity and symmetry.
\end{itemize}

As we show in \cref{thm:approxcomp}, the \emph{(lab)}, \emph{(exp)}, and \emph{(weak)} rules together already yield a system which is (approximately) complete.
To make our proof system and its presentation more pleasant, we include two additional rules: a \emph{zero} rule; and a \emph{symmetry} rule.
The \emph{(zero)} rule can be derived as an instance of the \emph{(exp)} rule with \(h\) constant (and thus generating no proof obligations).
Further, we can recursively replace any application of the \emph{(symm)} rule by the symmetric version of the proof.
The \emph{(zero)}, \emph{(symm)}, and \emph{(weak)} rules are inspired by those from quantitative equational theories~\cite{DBLP:conf/lics/MardarePP16}.
Together, this brings us to the following system:
\begin{definition}\label{def:proofrules}
	Let \((X,L,\tau,l)\) be an LMC, \(x,y \in X\), and \(\varepsilon \in \UIQ\). Further, we define \(\supps_{x,y} := \{ s \in X \mid \tau(x)(s) \neq \tau(y)(s) \}\) and \(\mu \cdot_{\supps_{x,y}} h := \sum_{s \in {\supps_{x,y}}} \mu(s) \cdot h(s)\).

	Then, we define the following derivation rules:
	\begin{equation*}
		\AxiomC{\vphantom{\(y \apart{\varepsilon} x\)}}
		\RightLabel{(zero)}
		\UnaryInfC{\(x \apart{0} y\)}
		\DisplayProof
		\qquad
		\AxiomC{\(y \apart{\varepsilon} x\)}
		\RightLabel{(symm)}
		\UnaryInfC{\(x \apart{\varepsilon} y\)}
		\DisplayProof
		\qquad
		\AxiomC{\(x \apart{\varepsilon'} y\)}
		\AxiomC{\(\varepsilon \leq \varepsilon'\)}
		\RightLabel{(weak)}
		\BinaryInfC{\(x \apart{\varepsilon} y\)}
		\DisplayProof
		\qquad
		\AxiomC{\(l(x) \neq l(y)\)}
		\RightLabel{(lab)}
		\UnaryInfC{\(x \apart{1} y\)}
		\DisplayProof
	\end{equation*}
	\begin{equation*}
		\AxiomC{\(h \colon \supps_{x,y} \to \UIQ\)}
		\noLine
		\UnaryInfC{\(\forall x', y' \in \supps_{x,y} \ldotp h(x') > h(y') \implies x' \apart{h(x') - h(y')} y'\)}
		\AxiomC{\((\tau(x) - \tau(y)) \cdot_{\supps_{x,y}} h \geq \varepsilon\)}
		\RightLabel{(exp)}
		\BinaryInfC{\(x \apart{\varepsilon} y\)}
		\DisplayProof
	\end{equation*}
	We may drop the subscripts \(x,y\) and \(\supps\) whenever clear from the context.
	We write \(\PTs{X}\) for the smallest set which contains all instances of the \emph{(zero)} and \emph{(lab)} rules for \(x,y \in X\), and is closed under applications of all instances of \emph{(symm)}, \emph{(weak)}, and \emph{(exp)} for any \(x,y \in X\) and \(\varepsilon \in \UIQ\).
\end{definition}
Note that we write conditions as premises in the rules rather than as separate side-conditions for convenience. As suggested by the definition of \(\PTs{X}\), this makes both the \emph{(zero)} and \emph{(lab)} rules axioms.
Further, strictly speaking, \emph{(exp)} is a family of rules, indexed by the maps \(h\).
The set \(\PTs{X}\) can be thought of as the set of all proof trees which can be built from the given rules.
As is usual, we will write \(\vdash x \apart{\varepsilon} y\) to mean that the given judgment is provable, i.e., there is a proof tree in \(\PTs{X}\) with the given judgment at the root.
\begin{remark}
	Note that the restriction to \(\supps\) in \emph{(exp)} means proofs, which are finite depth by definition, will also be finite breadth even when the LMC under consideration has an infinite state space.
	This is because \(\{ s \in X \mid \tau(x)(s) \neq \tau(y)(s) \}\) is a subset of \(\supp(\tau(x)) \cup \supp(\tau(y))\), which is finite by assumption.
	We will see later that this allows us to provide witnesses as both finite proof trees and finite modal formulas.
	This improves on the earlier work of~\cite{DBLP:conf/fossacs/RadyB23} which restricts to finite state spaces.
	We illustrate this improvement in \cref{ex:randwalk}, once we have shown soundness and completeness of the proof system, and its correspondence with modal formulas.
\end{remark}

\begin{example}\label{ex:simpleproof}
	We continue with the LMC from \cref{ex:LMC} and show how we can prove the distance between \(x\) and \(y\) shown there as a lower bound. Using the \emph{(lab)} rule, the bounds \(u \apart{1} v\) can be proved for \(u \in \{x_1,y_1\}\) and \(v \in \{x_2,y_2\}\).
	These proof obligations are generated by the map \(h_0 \colon \supps \to \UIQ\) defined by \(h_0(z) = \textbf{if } z \in \{x_1,y_1\} \textbf{ then } 1 \textbf{ else } 0\) for which \((\tau(x) - \tau(y)) \cdot h_0 = \frac{1}{10}\).
	We can thus prove:
	\begin{prooftree}%
		\AxiomC{}
		\UnaryInfC{\(x_1 \apart{1} x_2\)}
		\AxiomC{}
		\UnaryInfC{\(x_1 \apart{1} y_2\)}
		\AxiomC{}
		\UnaryInfC{\(y_1 \apart{1} x_2\)}
		\AxiomC{}
		\UnaryInfC{\(y_1 \apart{1} y_2\)}
		\QuaternaryInfC{\(x \apart{\frac{1}{10}} y\)}
	\end{prooftree}
\end{example}

\begin{example}\label{ex:noform}
	Our next example serves to illustrate a limitation of our proof system, namely that the behavioural distance of states will not always be exactly provable in our system. It would only be provable if we allowed infinite depth proof trees.
	The LMC we consider is the same as the one in~\cite[Thm.~17]{DBLP:conf/fossacs/RadyB23}, which shows that there is an LMC containing states for which it is not possible to give a single formula ``explaining'' their distance.
	
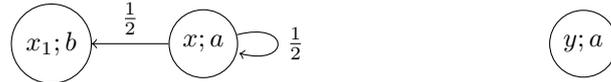
\begin{center}
		\begin{tikzpicture}[shorten >=0pt,node distance=2cm,
				on grid,auto,initial text=,every edge/.append style={bend angle=15}]
			\node[state] (x) {\slab{x}{a}} ;
			\node[state] (x1) [left=of x] {\slab{x_1}{b}} ;
			\node[state] (y) [right=5cm of x] {\slab{y}{a}} ;

			\path[->]
			(x) edge node[swap] {\(\frac{1}{2}\)} (x1)
			(x) edge[loop right] node[swap] {\(\frac{1}{2}\)} (x)
			;
		\end{tikzpicture}
	\end{center}
	As is discussed in \emph{op.\ cit.}, the distance \(\bd(x,y) = 1\) is reached only in the limit, not by any \(\Gamma^i(\bot)\). 
	We can see (by induction) that for any proved bound \(x \apart{\varepsilon} y\), there must exist some \(\Gamma^i(\bot)\) such that \(\varepsilon \leq \Gamma^i(\bot)(x,y)\) so that indeed no finite proof of \(x \apart{1} y\) can be given for this example.

	Proving the bound given by \(\Gamma^i(\bot)\) (for \(i > 0\)) can be done using \(i - 1\) applications of the \emph{(exp)} rule together with two applications each of the \emph{(lab)} and \emph{(zero)} rules.
	For example, once we have proved \(x_1 \apart{1} u\) for \(u \in \{x,y\}\), we can take the map \(h_0 \colon x,y \mapsto 0, x_1 \mapsto 1\) for which \((\tau(x) - \tau(y)) \cdot h_0 = \frac{1}{2}\) and prove:
	\begin{prooftree}
		\AxiomC{}
		\UnaryInfC{\(x_1 \apart{1} x\)}
		\AxiomC{}
		\UnaryInfC{\(x_1 \apart{1} y\)}
		\BinaryInfC{\(x \apart{\frac{1}{2}} y\)}
	\end{prooftree}
	This step (plus an application of \emph{(symm)}) allows the next application of \emph{(exp)} with a non-expansive \(h_0\) mapping \(x\) to \(\frac{1}{2}\), yielding a bound of \(\frac{3}{4}\). Continuing to increase the value of \(h_0(x)\) in this way, we approach \(\bd(x,y)\) from below.
\end{example}

\begin{example}
In the previous two examples, our strategy to compute the bound $\Gamma^{i+1}(\bot) (x,y)$ (for some pair $x,y$ of states) was to consider the function $h$ in the \emph{(exp)} rule to be the function of type $\Gamma^i(\bot)(\_,x')$ (for some $x'$ in the support of the distributions $\tau(x)$ or $\tau(y)$). In this example, we show that this strategy is \emph{incorrect} in general.

Consider the following LMC:

\begin{center}
		\begin{tikzpicture}[shorten >=0pt,node distance=2cm,
				on grid,auto,initial text=,every edge/.append style={bend angle=15}]
			\node[state] (x) {\slab{x}{a}} ;
			\node[state] (x1) [left=of x] {\slab{x_1}{a}} ;
			\node[state] (x2) [right=of x] {\slab{x_2}{b}} ;
			\node[state] (y) [right=7cm of x] {\slab{y}{a}} ;
			\node[state] (y1) [left=of y] {\slab{y_1}{c}} ;
			\node[state] (y2) [right=of y] {\slab{y_2}{d}} ;

			\path[->]
			(x) edge node[swap] {\(\frac{1}{2}\)} (x1)
			(x) edge node {\(\frac{1}{2}\)} (x2)
			(y) edge node[swap] {\(\frac{1}{2}\)} (y1)
			(y) edge node {\(\frac{1}{2}\)} (y2)
			;
		\end{tikzpicture}
\end{center}
We claim that $x \apart 1 y$. To see this, first note that the states $\supps_{x,y}=\{x_1,x_2,y_1,y_2\}$ are pairwise $1$-apart, i.e. $z \apart 1 z'$ for distinct $z,z' \in \supps_{x,y}$. Define a function $h\colon \supps_{x,y} \to [0,1]$ as follows: $x_1,x_2 \mapsto 1$ and $y_1,y_2 \mapsto 0$. Then using \emph{(emp)} rule we get $x \apart 1 y$ because
\[ \tau (x) \cdot h - \tau(y)\cdot h = (\frac 1 2 + \frac 1 2) - (0 + 0) = 1.\]
On the contrary, none of the following four functions allow us to prove exact behaviour distance between $x$ and $y$.
\begin{enumerate}
  \item Let $h_1 = \Gamma^1(\bot)(\_,x_1)$. Then we find that $h_1$ is inapplicable because $\tau (x) \cdot h_1 - \tau(y)\cdot h_2 = -\frac 1 2 <0$.
  \item Let $h_2 = \Gamma^1(\bot)(\_,x_2)$. Similar to the previous case.
  \item Let $h_3 = \Gamma^1(\bot)(\_,y_1)$. Then we find that $\tau (x) \cdot h_3 - \tau(y)\cdot h_3 = (\frac 1 2 + \frac 1 2) - (0 + \frac 1 2) = \frac 1 2$. Thus, we can only show $x \apart {\frac 1 2} y$ with $h_3$.
  \item Let $h_4 = \Gamma^1(\bot)(\_,y_2)$. Similar to the previous case.
\end{enumerate}
\end{example}

\subsection{Soundness}\label{sec:sound}
We now move on to showing soundness of the system, i.e., that if we can prove \(x \apart{\varepsilon} y\), then \(\bd(x,y) \geq \varepsilon\).
The \emph{(zero)} rule is sound as our pseudometrics are valued in \([0,1]\) and thus 0 is always a sound lower bound.
Similarly, the behavioural distance is symmetric, so that the order of states does not change a lower bound's validity.
Our weakening rule is sound by transitivity of \(\leq\).
Soundness of the label rule follows from the definition of \(\Gamma\). This is similar for the expectation \emph{(exp)} rule, as we discussed at the beginning of this section. We make the intuition given there precise in the coming results.

Due to our restriction of the domain of the map in the \emph{(exp)} rule, we will require the following lemma in the soundness proof:
\begin{restatable}{lemma}{supsupp}\label{lem:supsupp}
	For \((X,L,\tau,l)\) an LMC, \(d \colon X \times X \to [0,1]\) a pseudometric and \(x,y \in X\):
	\[
		\sup_{h \colon (X,d) \to ([0,1],d_e)} (\tau(x) - \tau(y)) \cdot h = \sup_{h \colon (\supps,d|_\supps) \to ([0,1],d_e)} (\tau(x) - \tau(y)) \cdot_\supps h
	\]
	where \(d|_\supps = d \circ (\iota_\supps \times \iota_\supps)\) with \(\iota_\supps \colon \supps \hookrightarrow X\) the inclusion map.
\end{restatable}
\begin{proofappendix}{supsupp}
	We prove two inequalities:

	\(\leq\):
	This holds because any \(h \colon (X,d) \to ([0,1],d_e)\) restricts to a map \(h|_\supps = h \circ \iota_\supps \colon (\supps,d|_\supps) \to ([0,1],d_e)\), and from the definition of \(\supps\) it follows that
	\[
		(\tau(x) - \tau(y)) \cdot h = (\tau(x) - \tau(y)) \cdot_\supps h|_\supps
	\]

	\(\geq\):
	For this direction, we show that for any \(h \colon (\supps, d|_\supps) \to ([0,1],d_e)\), there is an \(h' \colon (X,d) \to ([0,1],d_e)\) such that
	\[
		(\tau(x) - \tau(y)) \cdot h' = (\tau(x) - \tau(y)) \cdot_\supps h
	\]
	We use an existing construction of extensions of non-expansive maps (see \cite{mcshane1934extension,whitney1934analytic}), to extend \(h\) along the inclusion \(\iota_\supps \colon \supps \hookrightarrow X\). Namely, we define \(h'(x) := \inf_{z \in \supps} h(z) \oplus d(x,z)\), where \(\oplus\) is truncated addition on the unit interval.
	This is an extension in the sense that \(h' \circ \iota_\supps = h\), so that the above equality indeed holds.
\end{proofappendix}

We are now able to show (by structural induction) that any proof in our system yields a lower bound on \(\bd\).
\begin{restatable}[Soundness]{theorem}{sound}\label{lem:sound}
	For any LMC \((X,L,\tau,l)\), any proof tree built from the rules of \cref{def:proofrules}, any \(\varepsilon \in \UIQ\), and any \(x, y \in X\), if the proof tree has \(x \apart{\varepsilon} y\) at the root, then \(\bd(x,y) \geq \varepsilon\).
\end{restatable}
\begin{proofappendix}{sound}
	We proceed by induction on the structure of the proof tree.

	\textbf{Case \emph{(zero)}:} By its definition, \(\bd\) takes values in \([0,1]\), so that \(\bd(x,y) \geq 0\) holds.

	\textbf{Case \emph{(label)}:} We have a proof tree
	\[
		\AxiomC{\(l(x) \neq l(y)\)}
		\UnaryInfC{\(x \apart{1} y\)}
		\DisplayProof
	\]
	By definition of \(\Gamma\), we must have \(\bd(x,y) = 1\), so that indeed \(\bd(x,y) \geq 1\).

	\textbf{Case \emph{(symm)}:} We have a proof tree
	\[
		\AxiomC{\(y \apart{\varepsilon} x\)}
		\UnaryInfC{\(x \apart{\varepsilon} y\)}
		\DisplayProof
	\]
	By induction, we have \(\bd(y,x) \geq \varepsilon\), but \(\bd\) is symmetric, so that also \(\bd(x,y) \geq \varepsilon\).

	\textbf{Case \emph{(weak)}:} We have a proof tree
	\[
		\AxiomC{\(x \apart{\varepsilon'} y\)}
		\AxiomC{\(\varepsilon \leq \varepsilon'\)}
		\BinaryInfC{\(x \apart{\varepsilon} y\)}
		\DisplayProof
	\]
	By induction, we have \(\bd(x,y) \geq \varepsilon' \geq \varepsilon\).

	\textbf{Case \emph{(exp)}:} We have a proof tree
	\begin{equation*}
		\expr
	\end{equation*}
	By induction, we have for all \(x',y' \in \supps\) with \(h(x') > h(y')\) that \(\bd(x',y') \geq |h(x') - h(y')|\). For \(x',y' \in \supps\) with \(h(x') < h(y')\), we have \(\bd(x',y') = \bd(y',x') \geq |h(y') - h(x')| = |h(x') - h(y')|\). For the remaining pairs, \(|h(x') - h(y')| = 0 \leq \bd(x',y')\). In other words, \(h \colon \supps \to \UIQ\) is a non-expansive map \(h \colon (\supps,\bd) \to (\UIQ,d_e)\).
	As \(\bd\) is defined as the least fixed point of \(\Gamma\), we have
	\begin{align*}
		\bd(x,y) & = \Gamma(\bd)(x,y)                                                                                                              \\
		         & = \begin{cases}
			             1,                                                                             & \text{ if } l(x) \neq l(y), \\
			             \sup_{h \colon (X,\bd) \to ([0,1],d_e)} (\tau(x) - \tau(y)) \cdot h, & \text{ o.w.}
		             \end{cases}
	\end{align*}
	In case \(l(x) \neq l(y)\), we have \(\bd(x,y) = 1 \geq \varepsilon\).

	In the remaining case, we have
	\begin{align*}
		\bd(x,y) & = \sup_{k \colon (X,\bd) \to ([0,1],d_e)} (\tau(x) - \tau(y)) \cdot k                          \\
		         & = \sup_{h \colon (\supps,\bd|_\supps) \to ([0,1],d_e)} (\tau(x) - \tau(y)) \cdot_\supps h \\
		         & \geq (\tau(x) - \tau(y)) \cdot h \geq \varepsilon
	\end{align*}
	where the second equality is shown in \cref{lem:supsupp} and the first inequality holds because \(h\) is one of the non-expansive maps ranged over in the \(\sup\).
	This covers all cases, so soundness follows by induction.
\end{proofappendix}

\subsection{Approximate Completeness}\label{sec:appcomp}
We now go on to prove the approximate completeness of our system.
This will show that we can prove lower bounds arbitrarily close to the ``true'' value given by \(\bd\).
Our proof relies on the fact that we can get arbitrarily close to \(\bd\) with its finite approximants \(\Gamma^i(\bot)\) (\cref{prop:kleene}) and the following lemma, showing how to construct proofs exhibiting these finite approximants as lower bounds.

\begin{lemma}\label{lem:nstep}
	For any \(i \in \mathbb{N}\) and \(x,y \in X\), we have \(\vdash x \apart{\Gamma^i(\bot)(x,y)} y\).
\end{lemma}
\begin{proof}
	By induction on \(i\), where we strengthen the induction by showing that \(\Gamma^i(\bot)(x,y)\) is always rational.
	For the base case, we have \(\Gamma^0(\bot)(x,y) = \bot(x,y) = 0\) which is rational, and we can prove \(x \apart{0} y\) using the \emph{(zero)} rule.

	Now let \(i \in \mathbb{N}\), and suppose for any \(x,y \in X\), that we can prove \(x \apart{\Gamma^i(\bot)(x,y)} y\) and that \(\Gamma^i(\bot)(x,y)\) is rational. We have
	\begin{align*}
		\Gamma^{i+1}(\bot)(x,y) & = \Gamma(\Gamma^i(\bot))(x,y)                                                                                                     \\
		                        & = \begin{cases}
			                            1,                                                                                        & \text{ if } l(x) \neq l(y), \\
			                            \sup_{h \colon (X,\Gamma^i(\bot)) \to ([0,1],d_e)} (\tau(x) - \tau(y)) \cdot h, & \text{ o.w.}
		                            \end{cases}
	\end{align*}
	If \(l(x) \neq l(y)\), we can prove \(x \apart{1} y\) using the \emph{(lab)} rule.
	Otherwise, by \cref{lem:supsupp}, we have
	\begin{align*}
		\sup_{h \colon (X,\Gamma^i(\bot)) \to ([0,1],d_e)} (\tau(x) - \tau(y)) \cdot h = \sup_{h \colon (\supps,\Gamma^i(\bot)|_\supps) \to ([0,1],d_e)} (\tau(x) - \tau(y)) \cdot_\supps h
	\end{align*}
	By using \cref{thm:hfromLP},
	we can find a map in which the above supremum is achieved by solving a linear program.
	This will in particular be a map \(h_0 \colon (\supps,\Gamma^i(\bot)|_\supps) \to (\UIQ,d_e)\) (note the restriction to rationals) because all the coefficients in the problem are rational, by induction. We thus have
	\[
		\sup_{h \colon (\supps,\Gamma^i(\bot)|_\supps) \to ([0,1],d_e)} (\tau(x) - \tau(y)) \cdot_\supps h = (\tau(x) - \tau(y)) \cdot_\supps h_0
	\]
	which is rational, because the transition probabilities given by \(\tau(x)\) and \(\tau(y)\) are also rational.
	For each pair \(x', y'\) such that \(h_0(x') > h_0(y')\), the induction hypothesis yields a proof of \(x' \apart{\Gamma^{i}(\bot)(x',y')} y'\).
	By non-expansiveness of \(h_0\) with respect to \(\Gamma^{i}(\bot)\), we have \(h_0(x') - h_0(y') \leq \Gamma^{i}(\bot)(x',y')\).
	The proof obligations generated by \(h_0\) can thus be proved using the proofs given by the induction hypothesis, with an additional application of \emph{(weak)} in case \(h_0(x') - h_0(y')\) is strictly below \(\Gamma^{i}(\bot)(x',y')\).
	The following is thus a valid proof of the lower bound \(\varepsilon := \Gamma^{i+1}(\bot)(x,y)\):
	\begin{equation*}
		\AxiomC{\(\{x' \apart{h_0(x') - h_0(y')} y' \mid h_0(x') > h_0(y')\}\)}
		\AxiomC{\((\tau(x) - \tau(y)) \cdot h_0 \geq \varepsilon\)}
		\BinaryInfC{\(x \apart{\varepsilon} y\)}
		\DisplayProof
		\qedhere
	\end{equation*}
\end{proof}
We see that completeness in this sense only requires the use of the \emph{(zero)}, \emph{(lab)}, \emph{(exp)}, and \emph{(weak)} rules. In fact, it can be shown for only the latter three, due to the discussion preceding \cref{def:proofrules}. The approximate completeness of the system is now a simple consequence of the above lemma and \cref{prop:kleene}, with no further applications of the rules needed.
\begin{restatable}[Approximate Completeness]{theorem}{approxcomp}\label{thm:approxcomp}
	For any (real) \(\delta > 0\), and any \(x,y \in X\) there is a proof tree with \(x \apart{\varepsilon} y\) at the root, so that \(0 \leq \bd(x,y) - \varepsilon < \delta\).
\end{restatable}
\begin{proofappendix}{approxcomp}
	Let \(\delta > 0\). By \cref{prop:kleene}, we have \(\bd(x,y) = \sup_{i < \omega} \Gamma^i(\bot)(x,y)\), so there exists \(i \in \mathbb{N}\) such that
	\[
		0 \leq \bd(x,y) - \Gamma^i(\bot)(x,y) < \delta
	\]
	\Cref{lem:nstep} exactly gives us a proof of \(x \apart{\Gamma^i(\bot)(x,y)} y\), and we are done.
\end{proofappendix}

\subsection{Relation to two-valued apartness proofs}\label{sec:reltwovalapart}

As discussed in the introduction, when moving from bisimilarity/apartness to distances, we would like our new notions to be generalisations of existing two-valued definitions of (in)equivalence.
In this section, we show this by relating our new proof system for lower bounds to an inductive proof system for \emph{behavioural apartness}, developed in~\cite{DBLP:conf/cmcs/TurkenburgBKR24}.
Namely, Theorem 2 and Lemma 3 of \emph{op.~cit.} tell us that the following rule is sound and complete for \emph{behavioural apartness} (in this case, the dual of the usual probabilistic bisimilarity):
\begin{equation*}
	\AxiomC{\(\forall (x',y') \in R \ldotp x' \apart{} y'\)}
	\AxiomC{\(l(x) \neq l(y) \lor \exists z \in \supps \ldotp \tau(x)[z] \neq \tau(y)[z]\)}
	\BinaryInfC{\(x \apart{} y\)}
	\DisplayProof
\end{equation*}
where \(R \subseteq \supps_{xy} \times \supps_{xy}\), and we write \([z]\) for the equivalence class of \(z\) under the relation \(\eqcl(\overline{R^s})\) (equivalence closure of complement of symmetric closure of \(R\)).
Note that we make no further restriction on \(R\). In particular, \(R\) can be empty. This is the case for the \emph{(2-lab)} rule below, while in the \emph{(2-exp)} rule an empty relation will not distinguish any states, as \([z]\) will be the whole set \(\supps_{xy}\) for any chosen \(z\).

Note that we can write this as two separate rules
\begin{equation*}
	\AxiomC{\(l(x) \neq l(y)\)}
	\RightLabel{(2-lab)}
	\UnaryInfC{\(x \apart{} y\)}
	\DisplayProof
	\qquad
	\AxiomC{\(\forall (x',y') \in R \ldotp x' \apart{} y'\)}
	\AxiomC{\(\exists z \in \supps \ldotp \tau(x)[z] \neq \tau(y)[z]\)}
	\RightLabel{(2-exp)}
	\BinaryInfC{\(x \apart{} y\)}
	\DisplayProof
\end{equation*}
Just as we write \(x \apart{\varepsilon} y\) for provability in the proof system for lower bounds we have dealt with so far, we write \(\vdash x \apart{} y\) if there is a proof in the above system with \(x \apart{} y\) at the root.

The following result then shows the expected correspondence between such apartness proofs and those in our new proof system for lower bounds.
It can be seen as a dual of the well-known correspondence between probabilistic bisimilarity and having behavioural distance zero~\cite{DBLP:conf/concur/DesharnaisGJP99,DBLP:journals/tcs/FahrenbergL14}.
\begin{theorem}
	Let \((X, L, \tau, l)\) be an LMC, and \(x, y \in X\). Then we have
	\begin{equation*}
		\vdash x \apart{} y \iff \exists \varepsilon > 0 \ldotp \vdash x \apart{\varepsilon} y
	\end{equation*}
\end{theorem}
\begin{proof}
	\underline{\(\implies\)}: Suppose we have a proof of \(x \apart{} y\) for states \(x,y\) in an LMC \((X, L, \tau, l)\). We proceed by induction on the structure of this proof. In case the proof applies the \emph{(2-lab)} rule, we obtain a proof of \(x \apart{1} y\) using the \emph{(lab)} rule.

	Otherwise, the proof uses the \emph{(2-exp)} rule, from which we will construct an application of the \emph{(exp)} rule.
	We thus need to construct a map \(h\) which gives a non-zero lower bound, and show how to prove the proof obligations it generates.

	To this end, let \(R\) be the relation used in the application of the \emph{(2-exp)} rule. Note that \(R\) can not be empty, as we assume a valid proof of \(x \apart{} y\). By the induction hypothesis we have, for each \((x',y') \in R\), some \(\varepsilon_{x' y'} > 0\) and a proof of \(x' \apart{\varepsilon_{x' y'}} y'\).
	Now let
	\begin{equation*}
		\varepsilon := \min \{ \varepsilon_{x' y'} \mid (x',y') \in R \}
	\end{equation*}
	and let \(z \in \supps\) be such that \(\tau(x)[z] \neq \tau(y)[z]\).
	Then define \(h \colon \supps \to \UIQ\) such that for \(w \in \supps\):
	\begin{equation*}
		h(w) = \begin{cases}
			\varepsilon, & \text{ if } w \in [z], \\
			0, & \text{o.w.}
		\end{cases}
	\end{equation*}
	For any \(x',y' \in \supps\) with \(h(x') > h(y')\) we obtain a proof of \(x' \apart{h(x') - h(y')} y'\) either directly by the induction hypothesis in case \(h(x') - h(y') = \varepsilon\), or by applying \emph{(weak)} to the proof of the lower bound \(\varepsilon_{x' y'} \geq \varepsilon\) it provides (as \(\varepsilon\) was taken to be minimal).

	It now suffices to prove that \(\tau(x) \cdot h - \tau(y) \cdot h > 0\), which holds because
	\begin{align*}
		\tau(x) \cdot h - \tau(y) \cdot h &= \sum_{w \in [z]} h(w) \cdot (\tau(x)(w) - \tau(y)(w)) \\
		&= \sum_{w \in [z]} \varepsilon \cdot (\tau(x)(w) - \tau(y)(w)) \\
		&= \varepsilon \cdot (\tau(x)[z] - \tau(y)[z]) > 0
	\end{align*}

	An application of the \emph{(exp)} rule, thus yields a proof of the non-zero lower bound \(x \apart{\varepsilon} y\).

	\underline{\(\impliedby\)}: Let \(\varepsilon > 0\) such that there is a proof of \(x \apart{\varepsilon} y\) for some \(x,y\) states of an LMC \((X,L,\tau,l)\). We again apply induction to the structure of the proof. In case the proof applies the \emph{(lab)} rule, the \emph{(2-lab)} rule provides the required proof.
	For the case of the \emph{(symm)} rule, note that the \emph{(2-lab)} and \emph{(2-exp)} rules allow the states \(x\) and \(y\) to be swapped, yielding the required proof from the induction hypothesis.
	In case the \emph{(weak)} rule is used, the proof of \(x \apart{} y\) is given immediately by the induction hypothesis.

	For the \emph{(exp)} case define
	\begin{equation*}
		R := \{ (x',y') \mid h(x') > h(y') \}
	\end{equation*}
	where \(h \colon \supps \to \UIQ\) is the map used in the application of the \emph{(exp)} rule.
	To show that this \(R\) can be used in an application of the \emph{(2-exp)} rule, we must show that there is some \(z \in \supps\) witnessing a difference \(\tau(x)[z] \neq \tau(y)[z]\).
	Suppose, towards a contradiction, that for all \(z \in \supps\) we have \(\tau(x)[z] = \tau(y)[z]\).

	We can then rewrite the expected value in the \emph{(exp)} rule as follows:
	\begin{align*}
		(\tau(x) - \tau(y)) \cdot h &= \sum_{x' \in \supps} (\tau(x)(x') - \tau(y)(x')) \cdot h(x') \\
		&= \sum_{[z]} \sum_{x' \in [z]} (\tau(x)(x') - \tau(y)(x')) \cdot h(x') \\
		&\stackrel{(1)}{=} \sum_{[z]} h(z) \cdot \sum_{x' \in [z]} \tau(x)(x') - \tau(y)(x') \\
		&= \sum_{[z]} h(z) \cdot (\tau(x)[z] - \tau(y)[z]) = 0
	\end{align*}
	The equality (1) holds because \(h\) is constant on each class \([z]\) due to the definition of \(R\).
	We hereby reach a contradiction with the assumption that \(\varepsilon > 0\), so that there must exist \(z \in \supps\) with \(\tau(x)[z] \neq \tau(y)[z]\). We can thus apply the \emph{(2-exp)} rule with the \(R\) defined above to prove \(x \apart{} y\). 
\end{proof}

\section{Logic}\label{sec:logic}
The previous sections have given us a way to inductively derive lower bounds on the behavioural distance between states of an LMC, and shown soundness and approximate completeness of the proof system with respect to the behavioural distance \(\bd\).
In this sense, the system gives finite evidence or \emph{witnesses} for (lower bounds on) behavioural distances.

Another approach to giving such evidence is to construct formulas in some logic which, in the terminology of~\cite{DBLP:conf/fossacs/RadyB23}, ``explain'' the difference. Ideally, given states \(x,y \in X\), this would be a formula \(\varphi\) such that \(|\interp{\varphi}(x) - \interp{\varphi}(y)| = \bd(x,y)\), i.e., the difference in interpretations of \(\varphi\) on the states is exactly equal to their behavioural distance. However, as is shown in \emph{op.\ cit.}, such a formula can not be given in general (cf.\thinspace \cref{ex:noform}).
Instead, a construction is given of formulas corresponding to finite approximations of \(\bd\).
In our notation, these are formulas \(\varphi\) such that \(|\interp{\varphi}(x) - \interp{\varphi}(y)| = \Gamma^i(\bot)(x,y)\) (for some \(i\in\mathbb N\)).
As discussed in \cref{sec:behdist}, \(\bd\) can be obtained as the countable limit of these approximations, so that the construction of~\cite{DBLP:conf/fossacs/RadyB23} gives formulas explaining the behavioural distance of states up to an arbitrarily small error.

In this section, we give analogous constructions between proofs and formulas in the same logic used to characterise the behavioural distance \(\bd\) in~\cite{DBLP:conf/fossacs/RadyB23}.
We start by recalling this logic and its interpretation on LMCs.
Its semantics is given in terms of a real-valued interpretation function (first suggested in~\cite{DBLP:journals/jcss/Kozen85}) and is a slight variation of the logic studied in relation to behavioural distances in~\cite{DBLP:journals/tcs/DesharnaisGJP04}.
We then move onto the related constructions, the first of which is a straightforward inductive construction of a proof that the distance \(|\interp{\varphi}(x) - \interp{\varphi}(y)|\) is a lower bound.
The second, again inductively, constructs a formula witnessing some proved lower bound. This is based on constructions in~\cite{DBLP:conf/fossacs/RadyB23} and relies on a non-trivial lemma in the case where the distance of states arises from the supremum case in the definition of \(\Gamma\).

\begin{definition}\label{def:logic}
	Define the syntax of the logic \(\Lnot\) by the following grammar:
	\begin{equation*}
		\varphi ::= a \mid \bigcirc \varphi \mid \lnot \varphi \mid \varphi \ominus q \mid \varphi \lor \varphi
	\end{equation*}
	where \(a \in L\) and \(q \in \UIQ\). Further, given an LMC \((X,L,\tau,l)\), the quantitative semantics of \(\Lnot\) is given by the interpretation function \(\interp{\cdot} \colon \Lnot \to X \to \UIQ\) defined recursively by the following equations:
	\begin{alignat*}{3}
		\interp{a} (x)                & = \begin{cases} 1, & \text{ if } l(x) = a, \\ 0, & \text{ o.w.} \end{cases} \qquad &  & \interp{\varphi \ominus q}(x)  &  & = \max(0, \interp{\varphi}(x) - q)              \\
		\interp{\bigcirc \varphi} (x) & = \tau(x) \cdot \interp{\varphi}                                                 &  & \interp{\varphi \lor \psi} (x) &  & = \max(\interp{\varphi} (x), \interp{\psi} (x)) \\
		\interp{\lnot \varphi} (x)    & = 1 - \interp{\varphi} (x)                                                         &  &                                &  &
	\end{alignat*}
\end{definition}

\begin{remark}
	From the connectives in the logic \(\Lnot\), it is possible to define also \(\land\) and \(\oplus\) as \(\varphi \land \psi := \lnot (\lnot \varphi \lor \lnot \psi)\) and \(\varphi \oplus q := \lnot (\lnot \varphi \ominus q)\), which then have the expected semantics.
	We also write \(\mathsf{false}\) for the formula \(a \ominus 1\) whose interpretation is everywhere zero.
\end{remark}

\begin{example}
	Consider the LMC from \cref{ex:noform}, and the formulas \(\varphi_i := \bigcirc^i b\) for \(i \in \mathbb{N}\). We can show that \(\interp{\varphi_i}(x_1) = 1\) for any \(i\), so that \(\interp{\varphi_i}(x) = \sum_{n = 1}^i \left(\frac{1}{2}\right)^{i}\), while \(\interp{\varphi_i}(y) = 0\).
	The formula \(\varphi_i\) captures the probability of reaching a state with label \(b\) after \(i\) steps.
\end{example}

The logic and its interpretation induce new distances between states, namely the difference in interpretations \(|\interp{\varphi}(x) - \interp{\varphi}(y)|\). Our first correspondence result shows that this distance can be shown to be a lower bound on \(\bd(x,y)\) by a proof in our system.
\begin{restatable}{theorem}{ftop}\label{thm:ftop}
	For any LMC \((X,L,\tau,l)\), formula \(\varphi \in \Lnot\), and \(x,y \in X\), there exists a proof tree with \(x \apart{\varepsilon} y\) as its root, where \(\varepsilon = |\interp{\varphi}(x) - \interp{\varphi}(y)|\).
\end{restatable}
\begin{proofappendix}{ftop}
	By induction on the structure of \(\varphi\), where we write \(\pt{\psi,x,y}\) for the proof tree constructed for a formula \(\psi\) and states \(x,y\).

	\textbf{Case \(\varphi = a\):} We have \(\varepsilon = 0\) or \(\varepsilon = 1\) with the following proofs:
	\begin{equation*}
		\zero \qquad \lab
              \end{equation*}

	\textbf{Case \(\varphi = \lnot \psi\):} We have \(|\interp{\varphi}(x) - \interp{\varphi}(y)| = |\interp{\psi}(x) - \interp{\psi}(y)|\) so simply let \(\pt{\varphi,x,y} = \pt{\psi,x,y}\).

	\textbf{Case \(\varphi = \psi \ominus q\):} In this case we will have \(|\interp{\varphi}(x) - \interp{\varphi}(y)| \leq |\interp{\psi}(x) - \interp{\psi}(y)|\) (truncation may give an inequality) so that we have
	\begin{equation*}
		\AxiomC{\(x \apart{|\interp{\psi}(x) - \interp{\psi}(y)|} y\)}
		\AxiomC{\(|\interp{\varphi}(x) - \interp{\varphi}(y)| \leq |\interp{\psi}(x) - \interp{\psi}(y)|\)}
		\RightLabel{(weak)}
		\BinaryInfC{\(x \apart{|\interp{\varphi}(x) - \interp{\varphi}(y)|} y\)}
		\DisplayProof
	\end{equation*}

	\textbf{Case \(\varphi = \varphi_1 \lor \varphi_2\):} We have
	\begin{align*}
		\varepsilon & = |\interp{\varphi_1 \lor \varphi_2}(x) - \interp{\varphi_1 \lor \varphi_2}(y)|                             \\
		            & = |\max(\interp{\varphi_1}(x),\interp{\varphi_2}(x)) - \max(\interp{\varphi_1}(y),\interp{\varphi_2}(y))|   \\
		            & \leq \max(|\interp{\varphi_1}(x) - \interp{\varphi_1}(y)|, |\interp{\varphi_2}(x) - \interp{\varphi_2}(y)|)
	\end{align*}
	so that we take \(\pt{\varphi,x,y}\) to be
	\begin{equation*}
		\AxiomC{\(\pt{\varphi_i,x,y}\)}
		\AxiomC{\(\varepsilon \leq \varepsilon'\)}
		\RightLabel{(weak)}
		\BinaryInfC{\(x \apart{\varepsilon} y\)}
		\DisplayProof
	\end{equation*}
	with \(\varphi_i\) the formula yielding the above maximum, which we have called \(\varepsilon'\).

	\textbf{Case \(\varphi = \bigcirc \psi\):} 
	Then \(\varepsilon = | (\tau(x) - \tau(y)) \cdot \interp{\psi} |\) and by induction, we have for all \(x',y' \in \supps\) with \(\interp{\psi}(x') > \interp{\psi}(y')\) trees \(\pt{\psi,x',y'}\).
	Now we must distinguish two cases. If \(\tau(x) \cdot \interp{\psi} \geq \tau(y) \cdot \interp{\psi}\), we take \(h = \interp{\psi}|_\supps\) and construct
	\begin{equation*}
		\AxiomC{\(\interp{\psi}|_\supps \colon \supps \to \UIQ\)}
		\AxiomC{\(\{\pt{\psi,x',y'} \mid h(x') > h(y')\}\)}
		\AxiomC{\(\varepsilon = (\tau(x) - \tau(y)) \cdot \interp{\psi}\)}
		\RightLabel{(exp)}
		\TrinaryInfC{\(x \apart{\varepsilon} y\)}
		\DisplayProof
	\end{equation*}
	Otherwise, we take \(h = \interp{\lnot \psi}|_\supps\), and have
	\begin{equation*}
		\def\defaultHypSeparation{\hskip .2cm}
		\AxiomC{\(\interp{\lnot \psi}|_\supps \colon \supps \to \UIQ\)}
		\AxiomC{\(\{\pt{\psi,x',y'} \mid h(x') > h(y')\}\)}
		\AxiomC{\(\varepsilon = (\tau(x) - \tau(y)) \cdot \interp{\lnot \psi}\)}
		\noLine
		\UnaryInfC{\(= (\tau(y) - \tau(x)) \cdot \interp{\varphi}\)}
		\RightLabel{(exp)}
		\TrinaryInfC{\(x \apart{\varepsilon} y\)}
		\DisplayProof
		\qedhere
	\end{equation*}
\end{proofappendix}
\begin{remark}
	Each case of the construction increases the depth of the proof by at most one.
	In this sense, the proof does not grow unexpectedly compared to the formula we start with.
	However, due to the branching in the \emph{(exp)} rule, the number of rules we apply will be larger than the number of operators in a given formula in general.
	For an example, recall the LMC from \cref{ex:LMC}, and consider the formula \(\bigcirc a\). This has just two operators, but the proof generated by the above procedure will contain five recursive proof obligations generated by \(\interp{\bigcirc a}\) in the application of \emph{(exp)}, each requiring at least one rule application.
\end{remark}

\section{Constructing formulas from proofs}\label{sec:ftop}

To complete the correspondence between proofs and modal formulas, in this section we provide a construction going from a proof of \(x \apart{\varepsilon} y\), to a formula \(\varphi\) such that \(|\interp{\varphi}(x) - \interp{\varphi}(y)| = \varepsilon\).
In fact, we construct formulas whose interpretation on \(x\) is equal to the lower bound, and whose interpretation on \(y\) is zero. This is required for our recursive construction.

Our construction is inspired by that of~\cite{DBLP:conf/fossacs/RadyB23}, which relies on the following lemma:
\begin{lemma}\label{lem:notstoneweierstrass}
	Let \(f \colon X \to [0,1]\). If for any \(x,y \in X\), we have a function \(g_{xy} \colon X \to [0,1]\) such that \(g_{xy}(x) = f(x)\) and \(g_{xy}(y) = f(y)\)
	then \(f = \max_x \min_y g_{xy} = \min_x \max_y g_{xy}\).
\end{lemma}
A proof of a more general (continuous) version of this result can be found in~\cite[Lemma~A7.2]{ash2014real}, where it is used in the proof of the Stone-Weierstrass Theorem.

The usefulness of this lemma is perhaps not very intuitive; the idea is that given an application of the \emph{(exp)} rule, we may recursively assume formulas \(\varphi_{x' y'}\) capturing the distance between successor states \(x',y'\), i.e., \(\varphi_{x' y'}(x') = h(x') - h(y')\) and \(\varphi_{x' y'}(y') = 0\).
We thus have formulas whose interpretations match the function \(h\), but only for pairs of states, so that a lemma similar to the above is required to combine them into a single formula. 

Note that, due to the form of our \emph{(exp)} rule, we can not assume in an inductive proof that formulas are given for all successors; only those \(x',y'\) for which \(h(x') > h(y')\). It is possible to recover all other pairs using the \emph{(zero)} and \emph{(symm)} rules, but we wish to keep the formulas as small as we can. For this, we prove the following stronger version of \cref{lem:notstoneweierstrass}.
\begin{restatable}{lemma}{ordnotstoneweierstrass}\label{lem:ordnotstoneweierstrass}
	Let \(f \colon X \to [0,1]\). If for any \(x,y \in X\) such that \(f(x) \geq f(y)\), we have a function \(g_{xy} \colon X \to [0,1]\) such that:
	\begin{enumerate*}
		\item \(g_{xy}(x) = f(x)\)
		\item \(g_{xy}(y) = f(y)\)
		\item \(\forall z \in X \ldotp g_{xy}(z) \geq f(y)\)
		\item \(\forall z \in X \ldotp g_{xx}(z) = f(x)\)
	\end{enumerate*},
	then \(f = \max_{x} \min_{y : f(x) \geq f(y)} g_{xy}\).
\end{restatable}
\begin{proofappendix}{ordnotstoneweierstrass}
	We first define
	\begin{align*}
		k_{xy} & = \begin{cases}
			           g_{xy} & \text{if } f(x) \geq f(y) \\
			           g_{yx} & \text{if } f(x) < f(y)
		           \end{cases}
	\end{align*}
	The conditions of each case ensure that the required map \(g\) exists, and conditions (1) and (2) mean that the \(k_{xy}\) satisfy the requirements of \cref{lem:notstoneweierstrass}.
	We thus have \( f = \max_x \min_y k_{xy} \) and so it remains to prove that
	\begin{equation*}
		\max_x \min_y k_{xy} = \max_x \min_{y : f(x) \geq f(y)} g_{xy}
	\end{equation*}
	We prove this by proving two inequalities.

	\subparagraph*{\(\leq\):}
	Consider the inequality and simplify as follows:
	\begin{align*}
	& \forall z\in X\ldotp \max_x \min_y k_{xy}(z) \leq \max_x \min_{y\colon f(x) \geq f(y)} g_{xy}(z)\\
	\iff & \forall z\in X \ldotp \forall u_1 \in X \ldotp \min_y k_{u_1y}(z) \leq \max_x \min_{y\colon f(x) \geq f(y)} g_{xy}(z)\\
	\iff & \forall z\in X \ldotp \forall u_1 \in X \ldotp \exists u_3\in X\ldotp \min_y k_{u_1y}(z) \leq \min_{y\colon f(u_3) \geq f(y)} g_{u_3y}(z)\\
	\iff & \forall z\in X \ldotp \forall u_1 \in X \ldotp \exists u_3\in X\ldotp \forall u_4 \in X\ldotp f(u_3)\geq f(u_4) \implies \min_y k_{u_1y}(z) \leq g_{u_3u_4}(z)\\
	\iff & \forall z\in X \ldotp \forall u_1 \in X \ldotp \exists u_3\in X\ldotp \forall u_4 \in X\ldotp f(u_3)\geq f(u_4) \implies \exists u_2\in X \ldotp k_{u_1u_2}(z) \leq g_{u_3u_4}(z)
	\end{align*}

	So, let \(z, u_1 \in X\) and take \(u_3 = z\). Further, let \(u_4 \in X\) such that \(f(u_3) \geq f(u_4)\) and take \(u_2 = z\). Then
	\begin{align*}
		k_{u_1 u_2}(z) & = k_{u_1 z}(z) = f(z) \\
		g_{u_3 u_4}(z) & = g_{z u_4}(z) = f(z)
	\end{align*}
	This first inequality thus holds.

	\subparagraph*{\(\geq\):}
	Consider the inequality and simplify as follows:
	\begin{align*}
	& \forall z\in X\ldotp \max_x \min_y k_{xy} (z) \geq \max_x \min_{y : f(x) \geq f(y)} g_{xy} (z)\\
	\iff & \forall z\in X\ldotp \forall u_3\in X\ldotp \max_x \min_y k_{xy} (z) \geq \min_{y : f(u_3) \geq f(y)} g_{u_3y} (z)\\
	\iff & \forall z\in X\ldotp \forall u_3\in X\ldotp \exists u_1\in X \ldotp \min_y k_{u_1y} (z) \geq \min_{y : f(u_3) \geq f(y)} g_{u_3y} (z)\\
	\iff & \forall z\in X\ldotp \forall u_3\in X\ldotp \exists u_1\in X \ldotp \forall u_2 \in X\ldotp  k_{u_1u_2} (z) \geq \min_{y : f(u_3) \geq f(y)} g_{u_3y} (z)\\
	\iff & \forall z\in X\ldotp \forall u_3\in X\ldotp \exists u_1\in X \ldotp \forall u_2 \in X\ldotp  \exists u_4 \in X\ldotp f(u_3) \geq f(u_4) \land  k_{u_1u_2} (z) \geq  g_{u_3u_4 } (z).
	\end{align*}
	
	So, we let \(z, u_3 \in X\) and take \(u_1 = u_3\). Now let \(u_2 \in X\). We distinguish two further cases:
	\subparagraph*{\(f(u_2) > f(u_3)\):} Here we take \(u_4 = u_3\) and have
	\begin{align*}
		k_{u_1 u_2}(z) & = k_{u_3 u_2}(z) = g_{u_2 u_3}(z) \stackrel{(3)}{\geq} f(u_3) \stackrel{(4)}{=} g_{u_3 u_3}(z) = g_{u_3 u_4}(z)
	\end{align*}
	\subparagraph*{\(f(u_2) \leq f(u_3)\):} Here we take \(u_4 = u_2\) and see that
	\begin{align*}
		k_{u_1 u_2}(z) & = k_{u_3 u_2}(z) = g_{u_3 u_2}(z) = g_{u_3 u_4}(z)
	\end{align*}
	This concludes the case distinctions.
\end{proofappendix}

The proof is a little involved and not very informative, however it allows us to construct (potentially) smaller formulas witnessing distances as compared to~\cite{DBLP:conf/fossacs/RadyB23}.
This is because we do not require formulas distinguishing all pairs of reachable states in the construction of the next theorem.
We discuss the improvement in size at the end of this section, and illustrate it in \cref{ex:RvBBig,ex:randwalk}.

\begin{restatable}{theorem}{optptof}\label{thm:optptof}
	For any LMC \((X,L,\tau,l)\), any \(\varepsilon \in \UIQ\), and any \(x,y \in X\), if we have a proof of \(x \apart{\varepsilon} y\) using the rules of \cref{def:proofrules}, then there is a formula \(\varphi_{xy} \in \Lnot\) such that \(\interp{\varphi_{xy}}(x) = \varepsilon\) and \(\interp{\varphi_{xy}}(y) = 0\).
\end{restatable}
Note that we abuse notation, and use \(x \apart{\varepsilon} y\) to refer to both a judgment in a proof, and a proof tree with this judgment at the root.
\begin{proofappendix}{optptof}
	This is by induction on the structure of the proof.

	\textbf{Case \emph{(zero)}:} In this case we take \(\varphi_{xy} = \false\). We have \(\interp{\varphi_{xy}}(z) = 0\) for any \(z \in X\), yielding the desired interpretations.

	\textbf{Case \emph{(lab)}:} Here, we take \(\varphi_{xy} = l(x)\). We must have \(\interp{l(x)}(y) = 0\) as \(l(x) \neq l(y)\), so that the interpretations are as required.

	\textbf{Case \emph{(symm)}:} The induction hypothesis gives a formula \(\varphi_{yx}\). Taking \(\varphi_{xy} = \lnot \varphi_{yx} \ominus (1-\varepsilon)\), we have \(\interp{\varphi_{xy}}(x) = (1 - \interp{\varphi_{yx}}(x)) \ominus (1-\varepsilon) = \varepsilon\) and \(\interp{\varphi_{xy}}(y) = (1 - \interp{\varphi_{yx}}(y)) \ominus (1-\varepsilon) = 0\) as desired.

	\textbf{Case \emph{(weak)}:} The induction hypothesis here gives \(\varphi_{xy}'\) with \(\interp{\varphi_{xy}'}(x) = \varepsilon'\) and \(\interp{\varphi_{xy}'}(y) = 0\) and \(\varepsilon' \geq \varepsilon\). This means we can take \(\varphi_{xy} = \varphi_{xy}' \ominus (\varepsilon' - \varepsilon)\) and have \(\interp{\varphi_{xy}}(x) = \varepsilon' \ominus (\varepsilon' - \varepsilon) = \varepsilon\) and \(\interp{\varphi_{xy}}(y) = 0 \ominus (\varepsilon' - \varepsilon) = 0\).

	\textbf{Case \emph{(exp)}:} 
	Suppose we have a proof of the form
	\begin{equation*}
		\AxiomC{\(h \colon \supps \to \UIQ\)}
		\noLine
		\UnaryInfC{\(\forall x', y' \in \supps \ldotp h(x') > h(y') \implies x' \apart{h(x') - h(y')} y'\)}
		\AxiomC{\((\tau(x) - \tau(y)) \cdot h \geq \varepsilon\)}
		\BinaryInfC{\(x \apart{\varepsilon} y\)}
		\DisplayProof
	\end{equation*}
	By induction, we have formulas \(\varphi_{x'y'}\) such that \(\interp{\varphi_{x'y'}}(x') = h(x') - h(y')\) and \(\interp{\varphi_{x'y'}}(y') = 0\) for all \(x',y' \in \supps\) with \(h(x') > h(y')\). For those \(x',y'\) such that \(h(x') = h(y')\) we define \(\varphi_{x'y'} := \false\) (any formula which is everywhere zero can be used). Using these we construct, for \(x',y'\) such that \(h(x') \geq h(y')\), the formulas \(\psi_{x'y'}^h := \varphi_{x'y'} \oplus h(y')\).

	We now claim that the interpretations \(\interp{\psi_{x'y'}^h}\) satisfy the conditions of \cref{lem:ordnotstoneweierstrass}. The first two clearly hold.
	For the third, note that \(\interp{\varphi_{x'y'}}(z) \geq 0\) for any \(z\), so that indeed
	\begin{align*}
		\interp{\psi_{x'y'}^h}(z) & = \interp{\varphi_{x'y'}}(z) \oplus h(y') \geq 0 \oplus h(y') = h(y')
	\end{align*}
	For the fourth, we see that for any \(z \in X\):
	\begin{align*}
		\interp{\psi_{x'x'}^h}(z) & = \interp{\false \oplus h(x')}(z) = h(x')
	\end{align*}
	We now define
	\begin{align*}
		\varphi_{xy}^h := \bigvee_{x'} \left[ \left[ \bigwedge_{y' : h(x') > h(y')} \psi_{x'y'}^h \right] \land (\false \oplus h(x')) \right]
	\end{align*}
	This has the same interpretation as \(\bigvee_{x'} \bigwedge_{y' : h(x') \geq h(y')} \psi_{x'y'}^h\), because for pairs \((x',y')\) with \(h(x') = h(y')\), the formula \(\psi_{x'y'}^h\) will be equal to \(\false \oplus h(x')\) by definition.
	Note also that these formulas are finite, as we quantify over \(\supps\). 
	Thus letting \(\varphi_{xy} := \bigcirc \varphi_{xy}^h \ominus (\tau(y) \cdot h)\), yields a (finitary) formula with the desired property.
\end{proofappendix}

One may wonder why we have constructed the formulas \(\varphi_{xy}^h\) as a conjunction over a disjunction, and not vice versa. It turns out that this order matters in the case of our \emph{(exp)} rule, as the following example shows.
\begin{example}
	Consider the following LMC:
	
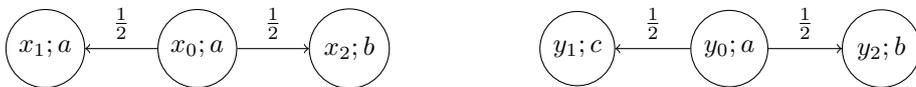
\begin{center}
		\begin{tikzpicture}[shorten >=0pt,node distance=2cm,
				on grid,auto,initial text=,every edge/.append style={bend angle=15}]
			\node[state] (x) {\slab{x_0}{a}} ;
			\node[state] (x1) [left=of x] {\slab{x_1}{a}} ;
			\node[state] (x2) [right=of x] {\slab{x_2}{b}} ;
			\node[state] (y) [right=7cm of x] {\slab{y_0}{a}} ;
			\node[state] (y1) [left=of y] {\slab{y_1}{c}} ;
			\node[state] (y2) [right=of y] {\slab{y_2}{b}} ;

			\path[->]
			(x) edge node[swap] {\(\frac{1}{2}\)} (x1)
			(x) edge node {\(\frac{1}{2}\)} (x2)
			(y) edge node[swap] {\(\frac{1}{2}\)} (y1)
			(y) edge node {\(\frac{1}{2}\)} (y2)
			;
		\end{tikzpicture}
	\end{center}
	For this example, we can prove the bound \(x_0 \apart{\frac{1}{2}} y_0\) using the \emph{(exp)} rule and the map \(h(x) = \textbf{if } x = x_1 \textbf{ then } 0 \text{ else } 1\), defined for \(x \in \{x_1,x_2,y_1,y_2\}\).
	The required recursive proofs generated by this \(h\) are \(u \apart{1} x_1\) for \(u \in \{x_2,y_1,y_2\}\), which can all be proved using the \emph{(lab)} rule.
	Constructing the \(\psi_{x'y'}^h\) for the pairs occurring in these recursive proofs yields
	\begin{align*}
		\psi_{x_2 x_1}^h = b \oplus 0 \quad \psi_{y_1 x_1}^h = c \oplus 0 \quad \psi_{y_2 x_1}^h = b \oplus 0
	\end{align*}
	We may now try to construct \(\varphi_{xy}^h\) as \(\bigwedge_{x'} \bigvee_{y : h(x) \geq h(y)} \psi_{x'y'}^h\).
	In the example, this gives \(\varphi_{x_0 y_0}^h \equiv \false\), i.e., the interpretation is zero everywhere thereby not matching the map \(h\).
	This also shows that we can not interchange the \(\max\) and \(\min\) in \cref{lem:ordnotstoneweierstrass}.
\end{example}

\subparagraph*{Size of constructed formulas}\label{sec:formsize}
As discussed in the introduction, our constructions together yield an algorithm going from an approximation \(\Gamma^i(\bot)(x,y)\) of the behavioural distance of states, via a proof tree, to a formula \(\varphi_{xy}\). This is an alternative to the construction given as an algorithm in~\cite[Sec.~7]{DBLP:conf/fossacs/RadyB23}.
In the worst case, this procedure will yield formulas whose size is exponential in the size of the corresponding LMC and the depth \(i\) of the approximation. We thus achieve the same asymptotic size complexity as the construction of \emph{op.\ cit.}, however there are large classes of examples for which the optimisations in our proof system lead to smaller formulas (when counting the total number of connectives).
The first example, taken from \emph{op.\ cit.}, will show a notable improvement in size, and will allow us to discuss the shapes of LMCs leading to these improvements.
Our final example applies to an LMC modelling random walks on the natural numbers. This is an infinite state example, which we are still able to capture as it is finitely branching.
\begin{example}\label{ex:RvBBig}
	We will compare the size of formulas obtained via our construction with those obtained in an example of~\cite[Sec.~5]{DBLP:conf/fossacs/RadyB23}. The LMC involved can be represented as follows:
	
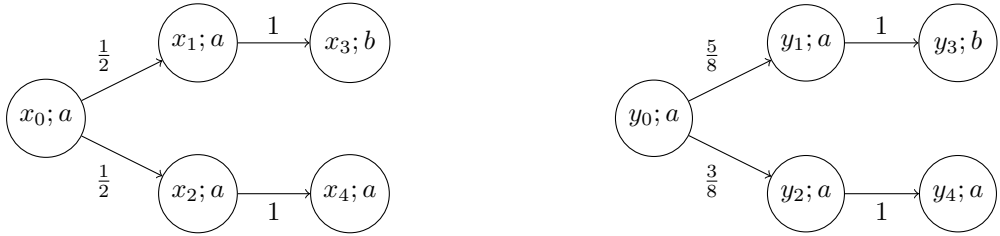
\begin{center}
		\begin{tikzpicture}[shorten >=0pt,node distance=2cm,
				on grid,auto,initial text=,every edge/.append style={bend angle=15}]
			\node[state] (x_0) {\slab{x_0}{a}} ;
			\node[state] (y_0) [right=8cm of x_0] {\slab{y_0}{a}} ;
			\node[state] (y_2) [below right=1cm and 2cm of y_0] {\slab{y_2}{a}} ;
			\node[state] (x_2) [below right=1cm and 2cm of x_0] {\slab{x_2}{a}} ;
			\node[state] (x_1) [above right=1cm and 2cm of x_0] {\slab{x_1}{a}} ;
			\node[state] (x_3) [right=of x_1] {\slab{x_3}{b}} ;
			\node[state] (x_4) [right=of x_2] {\slab{x_4}{a}} ;
			\node[state] (y_1) [above right=1cm and 2cm of y_0] {\slab{y_1}{a}} ;
			\node[state] (y_3) [right=of y_1] {\slab{y_3}{b}} ;
			\node[state] (y_4) [right=of y_2] {\slab{y_4}{a}} ;

			\path[->]
			(y_0) edge node[swap] {\(\frac{3}{8}\)} (y_2)
			(x_0) edge node[swap] {\(\frac{1}{2}\)} (x_2)
			(y_0) edge node {\(\frac{5}{8}\)} (y_1)
			(x_0) edge node {\(\frac{1}{2}\)} (x_1)
			(x_1) edge node {\(1\)} (x_3)
			(x_2) edge node[swap] {\(1\)} (x_4)
			(y_1) edge node {\(1\)} (y_3)
			(y_2) edge node[swap] {\(1\)} (y_4)
			;
		\end{tikzpicture}
	\end{center}
	It can be computed that \(\Gamma^3(\bot)(x_0,y_0) = \frac{1}{8}\). The construction applied in \cref{thm:approxcomp} gives the map \(h_0 \colon x_1, y_1 \mapsto 0, x_2, y_2 \mapsto 1\) for which \((\tau(x_0) - \tau(y_0)) \cdot h_0 = \frac{1}{8}\) so that we have:
	\begin{prooftree}
		\AxiomC{\(\vdots\)}
		\UnaryInfC{\(x_2 \apart{1} x_1\)}
		\AxiomC{\(\vdots\)}
		\UnaryInfC{\(x_2 \apart{1} y_1\)}
		\AxiomC{\(\vdots\)}
		\UnaryInfC{\(y_2 \apart{1} x_1\)}
		\AxiomC{\(\vdots\)}
		\UnaryInfC{\(y_2 \apart{1} y_1\)}
		\QuaternaryInfC{\(x_0 \apart{\frac{1}{8}} y_0\)}
	\end{prooftree}
	The recursive lower bounds, given by \(\Gamma^2(\bot)\), all have the same proof tree, up to renaming of states. For \(x_2 \apart{1} y_1\), we get \(h_0 \colon x_4 \mapsto 1, y_3 \mapsto 0\) giving \((\tau(x_4) - \tau(y_3)) \cdot h_0 = 1\) so that:
	\begin{prooftree}
		\AxiomC{}
		\UnaryInfC{\(x_4 \apart{1} y_3\)}
		\UnaryInfC{\(x_2 \apart{1} y_1\)}
	\end{prooftree}
	The formulas generated from such proofs are
	\begin{equation*}
		\varphi_{x_2 x_1} = \varphi_{x_2 x_1} = \varphi_{x_2 x_1} = \varphi_{x_2 x_1} = \bigcirc [[(a \oplus 0) \land (\false \oplus 1)] \lor [(\false \oplus 0)]] \ominus 0
	\end{equation*}
	Putting everything together, the formula \(\varphi_{x_0 y_0}\) is
	\begin{align*}
		\varphi_{x_0 y_0} = \bigcirc [ & [(\varphi_{x_2 x_1} \oplus 0) \land (\varphi_{x_2 y_1} \oplus 0) \land (\false \oplus 1)] \lor \\
		                               & [(\varphi_{y_2 x_1} \oplus 0) \land (\varphi_{y_2 y_1} \oplus 0) \land (\false \oplus 1)] \lor \\
		                               & [(\false \oplus 0)] \lor [(\false \oplus 0)]] \ominus \frac{3}{8}
	\end{align*}
	This has 8 recursive subformulas compared to the 100 occurring in the formula constructed in~\cite{DBLP:conf/fossacs/RadyB23} and could all be written out within around 5 lines. Clearly, the formula is still not minimal; it can be simplified to \(\bigcirc(\bigcirc a)) \ominus \frac{3}{8}\). However, we see a clear improvement in size.
\end{example}
The main features which allow us to achieve such an improvement in the size of formulas witnessing lower bounds are: the number of states reachable at each step being less than the size of the entire state space; and those successors having non-zero behavioural distance so that the map \(h\) takes many different values. The combination of restricting to supports and omitting symmetric pairs from the recursive proof obligations in the \emph{(exp)} rule, gives smaller proofs in these cases, which in turn are transformed into smaller formulas.

\begin{example}\label{ex:randwalk}
	We finish with an infinite state example based on random walks on the natural numbers. We model this as an LMC with state space \(\mathbb{N}\) and transitions \(\tau(n) = \frac{1}{2} \ket{n-1} + \frac{1}{2} \ket{n+1}\) for \(n > 0\) and \(\tau(0) = 1 \cdot \ket{0}\). Further, we have labels \(\{a,b\}\) and labelling function \(l(n) = \textbf{if } n = 0 \textbf{ then } b \textbf{ else } a\).

	States \(n < m\) can clearly be distinguished by the probability to reach the state \(0\) with unique label \(b\) in \(n\) steps.
	In fact, this turns out to completely determine the distance between states.
	For example, \(\Gamma^5(\bot)(4,6) = \frac{1}{2^4}\). This corresponds to the interpretation of the formula \(\bigcirc^4 b\) on these states.
	In general, we have for \(n < m\) and \(i > n\), \(\Gamma^i(\bot)(n,m) = \frac{1}{2^n}\).
	We will show the proof constructed for one such bound, as well as the formula constructed from this.

	We have \(\Gamma^3(\bot)(2,3) = \frac{1}{4}\) which can be proved as a lower bound using the map \(h_0 \colon 1 \mapsto \frac{1}{2}, 3 \mapsto 0, 2 \mapsto 0, 4 \mapsto 0\) for which \((\tau(2) - \tau(3)) \cdot h_0 = \frac{1}{4}\) as follows
	\vspace{-.5cm}
	\begin{prooftree}
		\AxiomC{\(\vdots\)}
		\UnaryInfC{\(1 \apart{\frac{1}{2}} 2\)}
		\AxiomC{\(\vdots\)}
		\UnaryInfC{\(1 \apart{\frac{1}{2}} 3\)}
		\AxiomC{\(\vdots\)}
		\UnaryInfC{\(1 \apart{\frac{1}{2}} 4\)}
		\TrinaryInfC{\(2 \apart{\frac{1}{4}} 3\)}
	\end{prooftree}
	The recursive proofs are all essentially the same, we show only the one for \(1 \apart{\frac{1}{2}} 2\), which uses the map \(h_0 \colon 0 \mapsto 1, 2 \mapsto 0, 1 \mapsto 0, 3 \mapsto 0\) yielding \((\tau(1) - \tau(2)) \cdot h_0 = \frac{1}{2}\) and
	\vspace{-.5cm}
	\begin{prooftree}
		\AxiomC{\(\vdots\)}
		\UnaryInfC{\(0 \apart{1} 1\)}
		\AxiomC{\(\vdots\)}
		\UnaryInfC{\(0 \apart{1} 2\)}
		\AxiomC{\(\vdots\)}
		\UnaryInfC{\(0 \apart{1} 3\)}
		\TrinaryInfC{\(1 \apart{\frac{1}{2}} 2\)}
	\end{prooftree}
	For these recursive proofs, the corresponding formulas are
	\begin{align*}
		\varphi_{12} = \varphi_{13} = \varphi_{14} = \bigcirc [ & [(b \oplus 0) \land (b \oplus 0) \land (b \oplus 0) \land (\false \oplus 1)] \lor    \\
		                                                        & [\false \oplus 0] \lor [\false \oplus 0] \lor [\false \oplus 0]] \ominus 0
	\end{align*}
	From these, we obtain
	\begin{align*}
		\varphi_{23} = \bigcirc \left[ \vphantom{\frac{1}{2}} \right. & \left[(\varphi_{12} \oplus 0) \land (\varphi_{13} \oplus 0) \land (\varphi_{14} \oplus 0) \land \left(\false \oplus \frac{1}{2}\right)\right] \lor \\
		                                                              & [\false \oplus 0] \lor [\false \oplus 0] \lor [\false \oplus 0]] \ominus 0
	\end{align*}
	This can be simplified to \(\bigcirc [\bigcirc b \land (\false \oplus \frac{1}{2})]\), which is not in general equivalent to \(\bigcirc \bigcirc b\), but does have the required property.
	We thus obtain evidence for differences in the behaviour of states even in a system with an infinite state space.
\end{example}

\section{Conclusions and Future Work}\label{sec:conclude}
We have given a derivation system for lower bounds on behavioural distances between states in labelled Markov chains, with proofs of soundness and approximate completeness with respect to a least fixed point definition of the behavioural distance.
The choice of the definition based on non-expansive maps was made specifically to allow the definition of a proof system, with the commonly used alternative definition based on couplings not immediately yielding a proof principle for lower bounds. The definitions are equivalent, by Kantorovich-Rubinstein duality~\cite{kantorovich1958space}. This duality arises more generally when defining equivalences and their apartness counterparts via liftings, as was noted in earlier work on \emph{behavioural apartness}~\cite{DBLP:conf/cmcs/TurkenburgBKR24}.\ruben{anonymised}
We further showed a close correspondence between proofs in our system and formulas in a modal logic, and compared this to the constructions in~\cite{DBLP:conf/fossacs/RadyB23} going between finite approximations of distances and formulas in the same logic. We see quite some avenues for future work, and sketch some of them here.

Definitions of behavioural distances have been given for a variety of other system types, both metric and probabilistic, e.g.: Metric LTSs~\cite{DBLP:conf/concur/Breugel05}; and Markov decision processes with finite~\cite{DBLP:conf/aaai/FernsPP04} and infinite state spaces~\cite{DBLP:conf/uai/FernsPP05}.
A natural extension would be
to generalise our results in case we change the system type while keeping a similar definition of distance between distributions; however we may also consider adapted notions of distance, such as the total variation distance studied for LMCs in~\cite{DBLP:conf/csl/ChenK14} or the MICo distance on MDPs~\cite{DBLP:conf/nips/CastroKPR21}.
Further statistical metrics/divergences such as the L\'evy-Prokhorov metric~\cite{prokhorov1956convergence}
or Kullback-Leibler divergence~\cite{10.1214/aoms/1177729694} may also be of interest.
The former has recently~\cite{DBLP:conf/csl/DesharnaisS26} been shown to define a functional on pseudometric spaces, whose greatest fixed point is the \(\varepsilon\)-distance~\cite{DBLP:conf/qest/DesharnaisLT08}.

A more general option is to take a coalgebraic view and use the definition of codensity lifting~\cite{DBLP:journals/lmcs/KatsumataSU18,DBLP:journals/logcom/SprungerKDH21} and its suitability for capturing quantitative notions of equivalence to provide a sound and complete derivation system for many of these systems at once.
Existing work on corresponding expressive logics~\cite{DBLP:conf/lics/KomoridaKKRH21} then gives a starting point for providing a general version of the construction in \cref{sec:ftop}.
Another approach in the same vein is to use the theory of Kantorovich functors~\cite{DBLP:conf/fossacs/GoncharovHNSW23}, used to obtain characteristic logics also in the quantitative setting.
We would also like to investigate the connection to strategies in quantitative bisimulation games studied in~\cite{DBLP:conf/concur/KonigM18,DBLP:conf/concur/VlasmanJ0B25} and developed for codensity bisimulations in~\cite{DBLP:journals/ngc/KomoridaKHKHEH22}.

The efficient computation of proofs and distinguishing formulas would also be an interesting extension.
There are a number of works in the qualitative setting (beyond the already discussed~\cite{DBLP:conf/fossacs/RadyB23}) which could provide inspiration.
For LTSs and branching bisimulation, computing (minimal) distinguishing formulas has been investigated by Martens and Groote~\cite{DBLP:conf/concur/0001G23,DBLP:conf/birthday/0001G24}.
K\"onig, Mika-Michalski and Schr\"oder use coalgebraic techniques to develop algorithms for computing strategies in bisimulation games and transforming these into distinguishing formulas~\cite{DBLP:conf/cmcs/0001MS20}.
Wi{\ss}mann, Milius and Schr\"oder give a coalgebraic algorithm related to partition refinement which constructs modal formulas characterising behavioural equivalence classes~\cite{DBLP:journals/lmcs/WissmannMS22}.

An alternative approach to improving the robustness of probabilistic bisimilarity, is the use of approximate bisimulations, such as: \(\varepsilon\)-bisimilarity~\cite{DBLP:conf/qest/DesharnaisLT08}; \(\varepsilon\)-APB~\cite{DBLP:conf/hybrid/DInnocenzoAK12}; robust probabilistic bisimilarity~\cite{DBLP:conf/cav/FatmiKPB25}; and \(\varepsilon\)-lumpability~\cite{buchholz1994exact}. These are often close to existing qualitative definitions, with some degree of error introduced. For a recent overview, and extension to weak and branching bisimulation, see~\cite{DBLP:conf/concur/SporkBKPQ24}. It would be interesting to compare these approximate notions to distances and relate them to proofs and logics.

\printbibliography

\newpage
\ifthenelse{\boolean{showappendix}}{%
\appendix
\section{Optimised expectation rule through linear programming}\label{sec:LProg}
Intuitively, a proof tree (rooted at $x\apart{\Gamma^i(x,y)} y$, corresponds to a finite approximant $\Gamma^i(x,y)$
of the behavioural distance $\bd(x,y)$ (we will make this connection formal in the next subsection). As a result,
the \emph{(exp)} rule operates inductively across these approximants: to establish a bound at step $i+1$, the rule
uses the pairwise state bounds $ \Gamma^i(\bot)(x', y')$ (where $x',y' \in \supps$) proven at step $i$ and searches for a
non-expansive map $h: S \to [0, 1]$ (w.r.t. the distance $\Gamma^i(\bot)|_S$ restricted to $\supps$) that maximises the
expected distance between transition  distributions $\tau(x)$ and $\tau(y)$.

Evaluating this inductive step requires computing a supremum over the uncountably infinite
function space  $[0,1]^\supps$ subject to non-expansiveness constraints. By framing this task as a linear program, we can replace this supremum by a maximum over the finite set of extreme points of a polytope
(bounded polyhedron) induced by the constraints encoding non-expansive functions.
Furthermore, this enables the use of standard LP solvers to compute the optimal evaluation function directly, thereby
reducing proof search to an efficient algorithmic optimisation problem. To state this connection formally, we start by
recalling some basic definitions from linear optimisation \cite{bertsimas1997introduction}.

\begin{definition}
A subset $P \subseteq \mathbb R^n$ is called a \emph{polyhedron} iff $P = \{\vec x \in \mathbb R^n \mid \mathbf A \cdot \vec x \leq  \mathbf b\}$,
where $\mathbf A$ is an $m\times n$-matrix and $\mathbf b \in \mathbb R^m$ is a vector.

A vector $\vec x \in P$ of a polygon $P \subseteq \mathbb R^n$ is a \emph{vertex} iff
\[
\exists \mathbf c.\quad \mathbf c \in \mathbb R^n  \ \land  \ \forall \vec y. \quad (\vec y\in P \land \vec y\neq \vec x) \implies \mathbf c^T \cdot \vec x < \mathbf c^T\cdot \vec y.
\]
A vector $\vec x \in P$ of a polygon $P \subseteq \mathbb R^n$ is an \emph{extreme point} iff $\vec x$ cannot be rewritten as a convex
sum of two points from $P$.
\end{definition}

It is well known (cf.\thinspace \cite[Theorem~2.3]{bertsimas1997introduction}) that the notions of vertex and extreme point are equivalent for any non-empty polyhedron.
Now by invoking a fundamental result (cf.\thinspace \cite[Theorem~2.6]{bertsimas1997introduction}) on the existence of extreme points,
we can further restrict the space of possible maps \(h\) in the \emph{(exp)} rule, by seeing the supremum as the optimal solution of a linear program $L(x,y,\supps,d)$ defined as follows.

We encode functions \(h \colon \supps \to [0,1]\) as (finite) vectors \(\vec{h} \in [0,1]^{|\supps|}\), writing \(\vec{h}_x\) for the value of \(\vec{h}\) at the position indexed by \(x \in \supps\).
	Then, each inequality \(|h(x) - h(y)| \leq d(x,y)\) can be expressed by \(\vec{a} \cdot \vec{h} \leq d(x,y)\) and \(\vec{a}' \cdot \vec{h} \leq d(x,y)\) with \(\vec{a}_x = 1, \vec{a}_y = -1, \vec{a}'_x = -1, \vec{a}'_y = 1\) (and all other entries zero). We can enforce \(0 \leq \vec{h}_x \leq 1\) similarly.
	We are thus interested in the problem of maximising \((\tau(x) - \tau(y)) \cdot h\) (a linear expression) subject to the constraints expressed by \(\mathbf{A} \cdot \vec{h} \leq \mathbf b\) for an integer matrix \(\mathbf{A}\) and vector \(\mathbf b\).

\begin{theorem}\label{thm:hfromLP}
Let \(\Hpoly\) be the feasible region of the linear program $L(x,y,\supps,d)$, i.e.
\[
\Hpoly = \{ \vec h \in [0,1]^{|\supps|} \mid \forall {s,t \in \supps}.\quad |\vec h_s - \vec h_t| \leq d(s,t)\}.
\]
Then the following statements hold.
\begin{enumerate}
\item The set $\Hpoly$ is convex, closed, and bounded; thus, a polytope.
\item The set $\Hpoly$ is nonempty.
\item The optimal value is achieved in the vertices of the polytope, which we denote by \(V(\Hpoly)\):
	\begin{equation*}
		\sup_{h \colon (X,d) \to ([0,1],d_e)} (\tau(x) - \tau(y)) \cdot h = \max_{h \in V(\Hpoly)} (\tau(x) - \tau(y)) \cdot_\supps h
	\end{equation*}
	The supremum becomes a maximum because finitely many linear inequalities define a polytope with finitely many vertices.
\end{enumerate}
\end{theorem}
\begin{proof}
The set $\Hpoly$ is convex and closed due to the shape of the constraints, and bounded due to the restriction to the unit interval. It is furthermore non-empty as any constant map valued in \([0,1]\) lies within it.

To prove the last item, we apply the well known Bauer's maximum principle\footnote{See \url{https://en.wikipedia.org/wiki/Bauer_maximum_principle}}: 
the maximum of a convex function $f\colon P \to \mathbb R$ is achieved at one of its vertices whenever $P$ is convex and compact. To this end, we need to show that our function $f\colon \Hpoly \to \mathbb R$ given by $h \mapsto \tau (x) \cdot h - \tau (y) \cdot h$ is convex. Note that $\tau(x) \cdot h$ simply denotes the expectation $\mathbb E_{\tau (x)}(h)$ of random variable $h$ under the distribution $\tau(x)$. Then using the fact that expectation is linear we derive:
\begin{align*}
    f(\alpha h + \beta k) =&\ (\tau(x)\cdot (\alpha h + \beta k) - \tau(y)\cdot (\alpha h + \beta k)) \\
    =&\ \mathbb E_{\tau(x)} (\alpha  h + \beta k) - \mathbb E_{\tau(y)} (\alpha h + \beta k)\\
    =&\ \big( \alpha \mathbb E_{\tau(x)} (h) + \beta \mathbb E_{\tau(x)} (k) \big) - \big( \alpha  \mathbb E_{\tau(y)} (h) + \beta \mathbb E_{\tau(y)} (k) \big)\\
    =&\ \alpha \big(\mathbb E_{\tau(x)}(h) - \mathbb E_{\tau (y)}(h) \big) + \beta \big(\mathbb E_{\tau(x)}(h) - \mathbb E_{\tau (y)}(k) \big)\\
    =&\ \alpha f(h) - \beta f(k)
\end{align*}
for any $\alpha,\beta$ such that $\alpha + \beta = 1$. Thus, $f$ is linear and therefore, the supremum over $\Hpoly$ reduces to a maximum over its finite vertices:
\[\sup_{\vec{h} \in \Hpoly} f(\vec{h}) = \max_{\vec{v} \in V(\Hpoly)} f(\vec{v}). \qedhere\]

\end{proof}
\ifthenelse{\boolean{proofsinappendix}}{%
	\section{Omitted Proofs}
	\closeoutputstream{proofstream}
	\input{\jobname-proofs.out}
}{}
}{}

\end{document}

%% file: defs.tex
\theoremstyle{plain}
\newtheorem{lemma}[thm]{Lemma}
\newtheorem{proposition}[thm]{Proposition}
\newtheorem{theorem}[thm]{Theorem}

\theoremstyle{definition}
\newtheorem{definition}[thm]{Definition}
\newtheorem{example}[thm]{Example}
\newtheorem{remark}[thm]{Remark}

\newcommand{\D}{\mathcal{D}}

\newcommand{\bd}{\mathsf{bd}}

\makeatletter
\providecommand{\bigsqcap}{%
  \mathop{%
    \mathpalette\@updown\bigsqcup
  }%
}
\newcommand*{\@updown}[2]{%
  \rotatebox[origin=c]{180}{$\m@th#1#2$}%
}
\providecommand{\plusop}{%
  \mathop{%
    \mathpalette\@updown\pm
  }%
}
\makeatother

\newcommand\PMet{\mathsf{PMet}}

\newcommand\supp{\mathrm{supp}}

\newcommand{\apart}[1]{\mathrel{\#_{#1}}}
\newcommand\subobj\rightarrowtail

\newcommand\fin\overline

\newcommand\paren[1]{\left({#1}\right)}
\newcommand\tuple\paren%

\DeclareMathOperator{\lfp}{lfp}

\newcommand{\interp}[1]{\llbracket #1 \rrbracket}
\newcommand{\UIQ}{[0,1]_\mathbb{Q}}
\newcommand{\Lnot}{\mathcal{L}}

\newcommand{\PTs}[1]{\mathscr{T}_{#1}}

\newcommand{\slab}[2]{\(#1;#2\)}
\newcommand{\supps}{\mathcal{S}}
\newcommand{\false}{\mathsf{false}}
\newcommand{\Hpoly}{H_{x,y}^d}
\DeclareMathOperator{\eqcl}{eq}

\newcommand{\ruben}[1]{\todo[size=\footnotesize]{ruben: #1}}

\newcommand{\zero}{%
\AxiomC{\vphantom{\(y \apart{\varepsilon} x\)}}
\RightLabel{(zero)}
\UnaryInfC{\(x \apart{0} y\)}
\DisplayProof%
}

\newcommand{\lab}{%
\AxiomC{\(l(x) \neq l(y)\)}
\RightLabel{(lab)}
\UnaryInfC{\(x \apart{1} y\)}
\DisplayProof%
}
\newcommand{\expr}{%
\AxiomC{\(h \colon \supps \to \UIQ\)}
\noLine
\UnaryInfC{\(\forall x', y' \in \supps \ldotp h(x') > h(y') \implies x' \apart{h(x') - h(y')} y'\)}
\AxiomC{\(\tau(x) \cdot h - \tau(y) \cdot h \geq \varepsilon\)}
\RightLabel{(exp)}
\BinaryInfC{\(x \apart{\varepsilon} y\)}
\DisplayProof%
}
\newcommand{\pt}[1]{\mathscr{T}(#1)}